\def\year{2018}\relax
\documentclass[letterpaper]{article} %
\usepackage{aaai18}  %
\usepackage{times}  %
\usepackage{helvet}  %
\usepackage{courier}  %
\usepackage{url}  %
\usepackage{graphicx}  %
\nocopyright

\usepackage{nohyperref}
\usepackage{amsmath,amsfonts,amssymb,amsthm}
\usepackage{txfonts}
\usepackage{graphicx}
\usepackage{xspace}
\usepackage{subfig}
\usepackage{booktabs} %
\usepackage{multirow}

\usepackage{float} %
\restylefloat{table}

\usepackage{tabularx}%

\usepackage{tikz}
\usetikzlibrary{arrows,automata}
\usetikzlibrary{arrows,shapes.geometric,positioning}
\usetikzlibrary{intersections}

\usepackage{fancyvrb} %

\usepackage{enumitem}

\def\PV{\mathcal{P}}  %
\def\alphPV{2^\PV}
\def\alphsym{\theta}

\def\trace{\ensuremath{\pi}} %

\def\lang{\mathcal{L}}

\def\translate{\mathsf{tr}}

\def\Cvars{\mathcal{Y}}
\def\Uvars{\mathcal{X}}

\def\strategy{\sigma}
\def\endtrace{\mathsf{end}}
\newcommand{\compat}[1]{\mathsf{ptraces}(#1)}
\newcommand{\fulltraces}[1]{\mathsf{traces}^\mathsf{f}(#1)} %
\newcommand{\tracestrat}[2]{\trace^\mathsf{f}[#1, #2]} %
\newcommand{\inftraces}[1]{\mathsf{traces}(#1)} %
\newcommand{\inftracestrat}[2]{\trace[#1, #2]}

\def\Fmc{\ensuremath{\mathcal{F}}} %

\newcommand{\eval}[3]{\llbracket #1, #2\rrbracket_{#3} }

\def\pvals{V}

\newcommand{\wcv}[1]{\mathsf{bgv}(#1)}  %
\newcommand{\wcvstart}[2]{\mathsf{bgv}_{#2}(#1)}
\def\opt{bgv-optimal\xspace} %

\def\desiredvals{\mathcal{V}}

\def\constraintdba{\C{A}^{\alpha_s, \alpha_c}_{\varphi}}

\def\cpredecessor{\mathsf{CPre}}

\newcommand{\showcomments}[1]{}

\newcommand{\acite}[1]{\citeauthor{#1} \shortcite{#1}}

\newcommand{\set}[1]{\left\{ #1 \right\}}
\newcommand{\la}{\langle}
\newcommand{\ra}{\rangle}

\newcommand{\act}{{\cal A}\xspace}

\def\my-isdef{\hbox{~$\stackrel{\text{def}}{=}$~}}

\renewcommand{\paragraph}[1]{\vspace{1mm}\noindent\textbf{#1}}

\newtheorem{theorem}{Theorem}
\newtheorem{lemma}{Lemma}

\newtheorem{proposition}{Proposition}

\theoremstyle{definition} 
\newtheorem{example}{Example}
\newtheorem{definition}{Definition}
\newtheorem{remark}{Remark}
\newcommand{\hide}[1]{}

\usepackage{colortbl}%
\usepackage{xcolor}%
\colorlet{tablerowcolor}{gray!10} %
\newcolumntype{L}[1]{>{\raggedright\arraybackslash}p{#1}}
\newcolumntype{C}[1]{>{\centering\arraybackslash}p{#1}}
\newcolumntype{R}[1]{>{@{\,}\raggedleft\arraybackslash}p{#1}@{\,}}

{\xspace}

\def\LTL{\ensuremath{\mathsf{LTL}}\xspace}
\def\LTLf{\ensuremath{\mathsf{LTL_{f}}}\xspace}

\def\LTLF{\ensuremath{\mathsf{LTL[\mathcal{F}]}}\xspace}
\def\LTLfF{\ensuremath{\mathsf{LTL_{f}[\mathcal{F}]}}\xspace}

\newcommand{\ltluntil}[2]{{#1} \mathcal{U} {#2}}

\newcommand{\ltlrelease}[2]{{#1} \mathcal{R} {#2}}
\newcommand{\ltlalways}[1]{\square{#1}}
\newcommand{\tikzcircle}[1]{\tikz[baseline=-0.5ex]\draw[black,radius=3pt,fill=#1] (0,0) circle ;}
\newcommand{\ltlweaknext}[1]{\tikzcircle{black}\xspace{#1}}
\newcommand{\ltlnext}[1]{\tikzcircle{white}\xspace{#1}}
\newcommand{\ltleventually}[1]{\lozenge\xspace{#1}}

\newcommand{\until}{\ensuremath{\operatorname{\mathsf{U}}}\xspace}
\newcommand{\release}{\ensuremath{\operatorname{\mathsf{R}}}\xspace}
\renewcommand{\ltluntil}[2]{{#1}\until{#2}}
\renewcommand{\ltlrelease}[2]{{#1}\release{#2}}

\newcommand{\assume}{\ensuremath{\alpha}}

\def\win{\mathsf{Reach}}
\newcommand{\constrainedproblem}[1]{\ensuremath{\la \xvars{#1}, \yvars{#1}, \assume{#1}, \varphi{#1} \ra}}
\newcommand{\agproblem}[2]{\ensuremath{\la \xvars, \yvars, #1, #2 \ra}}

\usepackage[textsize=scriptsize,textwidth=1cm]{todonotes}
\usepackage{bbm}

\newcommand{\C}[1]{\mathcal{#1}\xspace}

\newcommand{\NFA}{\text{NFA}\xspace} %
\newcommand{\DFA}{\text{DFA}\xspace} %

\newcommand{\NBA}{\text{NBA}\xspace}
\newcommand{\DBA}{\text{DBA}\xspace}

\def\defeq{\coloneqq}

\newcommand{\xvars}{\ensuremath{\C{X}}}
\newcommand{\yvars}{\ensuremath{\C{Y}}}
\newcommand{\synthproblem}[1]{\ensuremath{\la \xvars{#1}, \yvars{#1}, \varphi{#1} \ra}}

\newcommand{\QF}[1]{\ensuremath{F}_{#1}}

\def\actions{\C{A}}
\def\fluents{\C{F}}
\def\initform{\Psi_\mathit{init}}
\def\effform{\Psi_\mathit{eff}}
\def\oneform{\Psi_\mathit{one}}
\def\agentform{\Psi_\mathit{pre}}

\def\goalform{\Psi_\mathit{goal}}

\newcommand{\phione}{\ensuremath{ (\ltlnext{\top} \leftrightarrow \bigvee_{a\in\act{}} a) \land \bigwedge_{a,a'\in\act, a\neq a'} ( \neg a \vee \neg {a'}) } }

\begin{document}
\title{Finite LTL Synthesis with Environment Assumptions and Quality Measures}

\newcommand{\uoft}{$^\star$}
\newcommand{\umeg}{$^\dagger$}

\author{
Alberto Camacho\\
Department of Computer Science\\
University of Toronto, Canada\\
acamacho@cs.toronto.edu
\And
Meghyn Bienvenu\\
CNRS, Univ. Montpellier,  Inria\\
Montpellier, France \\
meghyn@lirmm.fr
\And
Sheila A. McIlraith\\
Department of Computer Science\\
University of Toronto, Canada\\
sheila@cs.toronto.edu
}

  \maketitle

\begin{abstract}
    In this paper, we investigate the problem of synthesizing strategies for linear temporal logic (\LTL) specifications that are interpreted over finite traces -- a problem that is central to the automated construction of controllers, robot programs, and business processes.
    We study a natural variant of the finite \LTL synthesis problem in which strategy guarantees are predicated on specified environment behavior. We further explore a quantitative extension of \LTL that supports specification of quality measures, utilizing it to synthesize high-quality strategies. We propose new notions of optimality and associated algorithms that yield strategies that best satisfy specified quality measures. Our algorithms utilize an automata-game approach, positioning them well for future implementation via existing state-of-the-art techniques.
\end{abstract}

\section{Introduction}
The problem of automatically synthesizing digital circuits from logical specifications was first proposed 
by \acite{Church57}.
In \citeyear{pnu-ros-popl89}, \citeauthor{pnu-ros-popl89}
examined the problem of synthesizing strategies for reactive systems, proposing Linear Temporal Logic (\LTL) \cite{pnueli77} as the specification language.
In a nutshell,  \LTL is used to express temporally extended properties of infinite state sequences (called traces), and the aim of \LTL synthesis is to produce a winning strategy, i.e. a function that assigns values to the state variables under the control of the system at every time step, in such a way that the induced infinite trace is guaranteed to satisfy the given \LTL formula, no matter how the environment sets the remaining state variables.  
In 2015, \citeauthor{deg-var-ijcai15} introduced the problem of $LTL_{f}$~\emph{synthesis} in which the specification is described in a variant of \LTL interpreted over finite traces \cite{DBLP:conf/ijcai/GiacomoV13}.  
Finite interpretations of \LTL have long been exploited to specify temporally extended goals and preferences in AI automated planning (e.g., \cite{BacchusK00,bai-bac-mci09}).
In contrast to \LTL\ synthesis, which produces programs that run in perpetuity, \LTLf synthesis is concerned with the generation of terminating programs.
Two natural and important application domains are automated synthesis of business processes, including web services; and automated synthesis of robot controllers, in cases where program termination is desired. 

Despite recent work on \LTLf synthesis, there is little written on the nature and form of the \LTLf specifications and how this relates to the successful and nontrivial realization of strategies for such specifications. 
\LTLf synthesis is conceived as a game between the environment and an agent. The logical specification that defines the problem must not only define the desired behavior that execution of the strategy should manifest -- what we might loosely think of as the \emph{objective} of the strategy, but must also define the context, including %
any \emph{assumptions} about %
the environment's behavior %
upon which realization of the objective is predicated. 
As we show in this work, if  %
assumptions about environment behavior are not appropriately taken into account, specifications can either be impossible to realize or can be realized trivially by allowing the agent to violate assumptions upon which guaranteed realization of the objective is predicated.

We  further examine
the problem of how to construct specifications where the realization of an objective comes with a quality measure, and where strategies provide guarantees with respect to these measures.  The addition of quality measures is practically motivated. In some instances we may have an objective that can be  realized in a variety of ways of differing quality (e.g., my automated travel assistant may find a myriad of ways for me to get to KR2018 -- some more preferable than others!). Similarly, 
we may have multiple objectives that are mutually unachievable and we may wish to associate a quality measure with their individual realization
(e.g., I'd like my home robot to do the laundry, wash dishes, and cook dinner before its battery dies, but dinner is most critical, followed by dishes).  

In this paper we explore finite \LTL synthesis with environment assumptions and quality guarantees.  In doing so, we uncover %
important 
observations regarding the form and nature of \LTLf synthesis specifications, how resulting strategies are computed, and the nature of the guarantees we can provide regarding the resulting strategies.  In Section \ref{sec:synth-env-ass} we examine the problem of \LTLf synthesis with environment assumptions, introducing the notion of \emph{constrained \LTLf synthesis} in Section \ref{sec:constrained}.  In Section \ref{sec:alg-cltlf-synth}, we propose algorithms for constrained \LTLf synthesis, including a 
reduction to Deterministic B\"{u}chi Automata games for the fragment of environment constraints that are conjunctions of safe and co-safe \LTL formulae.
In Section \ref{sec:high-quality-synth}, we examine the problem of augmenting constrained \LTLf synthesis with quality measures.  We adopt a specification language, \LTLfF, proposed by \cite{AlmagorKRV17} and define a new notion of optimal strategies.  In Section \ref{sec:alg-high-qual-synth}, we provide algorithms for computing high-quality strategies for constrained \LTLf synthesis.   %
Section \ref{sec:conclusion} summarizes our technical contributions.
Some proofs are deferred to the appendix. This paper (without the appendix proofs) will appear in the \emph{Proceedings of the 16th International Conference on Principles of Knowledge Representation and Reasoning} (KR 2018).
\section{Preliminaries}
 We recall the syntax and semantics of linear temporal logic for both infinite and finite traces,
as well as the basics of finite state automata and the link between \LTL\ and automata. 
\subsection{Linear Temporal Logic (\LTL)}
Given a set $\PV$ of propositional variables,  \LTL formulae are defined as follows:
\[
\varphi \defeq \top ~\vert~ \bot ~\vert~ 
p \mid \neg \varphi \mid \varphi_1 \land \varphi_2 \mid \varphi_1 \lor \varphi_2 \mid \ltlnext{} \varphi \mid \varphi_1 \ltluntil{}{} \varphi_2 \mid  \varphi_1 \ltlrelease{}{} \varphi_2
\]
where $p \in \PV$. Here $\neg$, $\wedge$, and $\vee$ are the usual Boolean connectives, and $\ltlnext{}$ (next), $\ltluntil{}{} $ (until), and $\ltlrelease{}{}$ (release) are temporal operators.
The formula $\ltlnext{} \varphi$ states that $\varphi$ must hold in the next timepoint, $\varphi_1 \ltluntil{}{} \varphi_2$ stipulates that $\varphi_1$
must hold until $\varphi_2$ becomes true, and $\varphi_1 \ltlrelease{}{} \varphi_2$ expresses that $\varphi_2$ remain true until and including the point in which $\varphi_1$ is made true (or forever if $\varphi_2$ never becomes true). %
For concision, we do not include logical implication ($\rightarrow$), eventually ($\ltleventually{}$, `sometime in the future') and always ($\ltlalways{}$, `at every point in the future') in the core syntax, but instead view them as abbreviations: 
$\alpha \rightarrow \beta \defeq \neg \alpha \lor \beta$,
$\ltleventually{} \varphi \defeq \top \ltluntil{}{} \varphi$ and
and $\ltlalways{} \varphi \defeq \bot \ltlrelease{}{} \varphi $. 

\LTL\ formulae are traditionally interpreted over \emph{infinite traces} $\trace$, i.e., infinite words 
over the alphabet $2^\PV$. 
Intuitively, an infinite trace $\trace$ describes an infinite %
sequence of (time)steps, with the $i$-th symbol in $\trace$, written $\trace(i)$, specifying the propositional symbols that hold at step~$i$.
We use $\trace \sqsubseteq \trace'$
to indicate that $\trace$ is a prefix of $\trace'$. 
We define what it means for an infinite trace $\trace$ to \emph{satisfy an \LTL formula $\varphi$ at step $i$}, denoted $\trace\models_i \varphi$: 
\begin{itemize}[noitemsep,topsep=0pt]
\item $\trace\models_i \top$, $\trace\not\models_i \bot$, and
$\trace \models_i p$ iff $p \in \trace(i)$, for each $p\in\PV$; 
\item $\trace\models_i \neg \varphi$ iff $\trace\not\models_i \varphi$;
\item $\trace\models_i \varphi_1\land\varphi_2$ iff $\trace\models_i \varphi_1$ and $\trace\models_i \varphi_2$;
\item $\trace\models_i \varphi_1\lor\varphi_2$ iff $\trace\models_i \varphi_1$ or $\trace\models_i \varphi_2$;
\item $\trace\models_i \ltlnext{} \varphi$ iff $\trace\models_{i+1}\varphi$;
\item $\trace\models_i \varphi_1\ltluntil{}\varphi_2$ iff there exists $j \geq i$ such that $\trace\models_j \varphi_2$, and for each $i \leq k < j$, $\trace\models_k\varphi_1$;
\item $\trace\models_i \varphi_1\ltlrelease{}\varphi_2$ iff for all $j \geq i$ either $\trace \models_j \varphi_2$ or there exists $i \leq k < j$ such that $\trace \models_k \varphi_1$. 
\end{itemize}
A formula $\varphi$ is \emph{satisfied in $\trace$}, written $\trace \models \varphi$, if $\trace\models_1 \varphi$. 
Two formulas $\varphi$ and $\psi$ are \emph{equivalent} if $\trace \models \varphi$ iff $\trace \models \psi$ for all traces $\pi$. 
Observe that, in addition to the usual Boolean equivalences, we have the following: 
$\varphi_1 \ltluntil{} \varphi_2 \equiv \neg (\neg \varphi_1 \ltlrelease{} \neg\varphi_2)$ and $\neg \ltlnext{} \varphi \equiv \ltlnext{} \neg \varphi$. 

We consider two well-known syntactic fragments of~\LTL.
The \emph{safe fragment} is defined as follows \cite{DBLP:journals/fac/Sistla94}:
\[
\varphi \defeq \top \mid  \bot \mid 
p \mid \neg p \mid  \varphi_1 \land \varphi_2  \mid \varphi_1 \vee \varphi_2  \mid \ltlnext{} \varphi \mid   \varphi_1 \ltlrelease{}{} \varphi_2
\]
The complementary \emph{co-safe fragment} is similarly defined, using
$\ltluntil{}{}$ in place of %
$\ltlrelease{}{}$. 
It is known that  if $\varphi$ is a safe formula
and $\trace \not \models \varphi$, then there is a finite \emph{bad prefix} $\trace_b \sqsubseteq \trace$ such that 
$\trace' \not \models \varphi$ for every infinite trace $\trace'$ with $\trace_b \sqsubseteq \trace'$. 
 Similarly, if $\varphi$ is a co-safe formula and $\trace \models \varphi$, then there exists a finite \emph{good prefix} $\trace_g \sqsubseteq \trace$
such that $\trace' \models \varphi$ for every infinite trace $\trace'$ with $\trace_g \sqsubseteq \trace'$.
This means that violation of safe formulae and satisfaction of co-safe formulae can be shown by exhibiting a suitable finite prefix \cite{DBLP:journals/fmsd/KupfermanV01}.  

In this paper, our main focus will be on a more recently studied \emph{finite version of \LTL}, denoted \LTLf 
\cite{DBLP:conf/ijcai/GiacomoV13}, in which formulae are interpreted over \emph{finite traces} (finite words over $2^\PV$). We will reuse the notation $\trace(i)$ ($i$-th symbol) %
and introduce the notation $|\trace|$ for the length of~$\trace$. %
\LTLf\ has precisely the same syntax as \LTL\ and the same semantics for the propositional constructs, but it differs in its interpretation of the temporal operators:
\begin{itemize}[noitemsep,topsep=0pt]
\item $\trace\models_i \ltlnext{} \varphi$ iff $|\trace|> i$ and $\trace\models_{i+1}\varphi$;
\item $\trace\models_i \varphi_1\ltluntil{}\varphi_2$ iff there exists $i \leq j \leq |\pi|$ such that $\trace\models_j \varphi_2$, and $\trace\models_k\varphi_1$, for each $i \leq k < j$;
\item $\trace\models_i \varphi_1\ltlrelease{}\varphi_2$ iff for all $i \leq j \leq |\pi|$ either $\trace \models_j \varphi_2$ or there exists $i \leq k < j$ such that $\trace \models_k \varphi_1$
\end{itemize}
We introduce %
the \emph{weak next} operator ($\ltlweaknext{}$) as an abbreviation: $\ltlweaknext{} \varphi \defeq \ltlnext{} \varphi \vee \neg \ltlnext{} \top$.
Thus, $\ltlweaknext{} \varphi $ holds if $\varphi$ holds in the next time step or we have reached the end of the trace. 
Over finite traces, $\neg \ltlnext{} \varphi \not \equiv \ltlnext{} \neg \varphi$, but we do have  $\neg \ltlweaknext{} \varphi \equiv \ltlnext{} \neg \varphi$. 

As before, we say that $\varphi$ is satisfied in $\trace$, written $\trace \models \varphi$,  if $\trace\models_1 \varphi$. Note that we can unambiguously use the same notation for \LTL\ and \LTLf\ so long as we specify whether the considered trace is finite or infinite.

\subsection{Finite State Automata}
We recall that a \emph{non-deterministic finite-state automaton (NFA)} is a tuple 
$\C{A} = \langle \Sigma, Q, \delta, Q_0, \QF{} \rangle$, where $\Sigma$ is a finite alphabet of \emph{input symbols}, $Q$ is a finite set of \emph{states}, $Q_0 \subseteq Q$ is a set of \emph{initial states}, $\QF{} \subseteq Q$ is a set of \emph{accepting states}, and $\delta: Q \times \Sigma \rightarrow 2^Q$ is the \emph{transition} function.
NFAs are evaluated on \emph{finite} words, i.e.\ elements of $\Sigma^*$. A \emph{run} of $\C{A}$ on a word $w = \theta_1 \cdots \theta_n$ is a sequence $q_0\cdots q_n$ of states, such that $q_0 \in Q_0$, and $q_{i+1} \in \delta(q_i,\theta_{i+1})$ for all $0 \leq i < n$. A run $q_0\cdots q_n$ is \emph{accepting} if $q_n \in \QF{}$, and $\C{A}$ \emph{accepts} $w$ if some run of $\C{A}$ on $w$ is accepting.
The \emph{language} of an automaton 
$\C{A}$, denoted $\C{L}(\C{A})$, 
is the set of words accepted by $\C{A}$.

\emph{Deterministic finite-state automata} (DFAs) are NFAs in which $|Q_0|=1$
and $|\delta(q,\theta)| = 1$ for all $(q,\sigma) \in Q \times \Sigma$. When $\C{A}$ is a \DFA, we write 
$\delta: Q \times \Sigma \rightarrow Q$ and $q' = \delta(q,\theta)$ in place of $q' \in \delta(q,\theta)$, 
and when $Q_0=\{q_0\}$, we will simply write $q_0$ (without the set notation). 
For every \NFA $\C{A}$, there exists a \DFA that accepts the same language as $\C{A}$ and whose size is at most single exponential in the size of $\C{A}$.
The \emph{powerset construction} is a well-known technique to \emph{determinize} NFAs \cite{DBLP:journals/ibmrd/RabinS59}.

\emph{Non-deterministic B\"uchi automata} (\NBA) are defined like NFAs but evaluated on \emph{infinite} words, that is, elements of~$\Sigma^\omega$. %
A \emph{run} of $\C{A}$ on an infinite word $w = \theta_1 \theta_2 \cdots$ is a sequence $\rho = q_0 q_1 q_2 \cdots$ of states, such that $q_0 \in Q_0$, and $q_{i+1} \in \delta(q_i,\theta_{i+1})$ for every $i \geq 0$.
A run $\rho$ is \emph{accepting} if $\mathsf{inf}(\rho) \cap \QF{} \neq \emptyset$, where $\mathsf{inf}(\rho)$ is the set of states that appear infinitely often in $\rho$. We say that an \NBA $\C{A}$ \emph{accepts} $w$ if some run of $\C{A}$ on $w$ is accepting.
Analogous definitions apply to \emph{deterministic} B\"uchi automata (DBAs).

We will also consider \emph{deterministic finite-state transducers} (also called Mealy machines, later abbreviated to `transducers'), given by tuples $\C{T} = \langle \Sigma, \Omega, Q, \delta, \omega, q_0 \rangle$, where $\Sigma$ and $\Omega$ are respectively the \emph{input and output alphabets}, 
$Q$ is the set of states, $\delta: Q \times \Sigma \rightarrow Q$ is the transition function, $\omega: Q \times \Sigma \rightarrow \Omega$ is the \emph{output function}, 
and $q_0$  the initial state. The run of $\C{T}$ on %
$w = \theta_1 \theta_2 \ldots \in \Sigma^\omega$
is an infinite sequence of states $q_0 q_1 q_2 \ldots$ with $q_{i+1} \in \delta(q_i,\theta_{i+1})$ for every $i \geq 0$,
and the \emph{output sequence} of $\C{T}$ on $w$ is $\omega(q_0, \theta_1) \omega(q_1,\theta_2) \ldots$.

Given an \LTLf formula $\varphi$, one can construct an \NFA
that accepts precisely those finite traces $\trace$ with $\trace \models \varphi$ (e.g. \cite{deg-var-ijcai15}).
For every safe formula $\varphi_s$ (resp.\ co-safe formula $\varphi_c$), one can construct an \NFA %
that accepts all bad prefixes of $\varphi_s$
(resp.\ good prefixes of $\varphi_c$) \cite{DBLP:journals/fmsd/KupfermanV01}. 
In these constructions, the \NFA\/s are 
worst case
single exponential in the size of the formula. 
By determinizing these \NFA\/s, we can obtain \DFA\/s of double-exponential size that recognize the same languages. %

\section{\LTL\ and \LTLf\ Synthesis}
\label{sec:synth-env-ass}
To set the stage for our work, we recall the definition of \LTL\ synthesis in the infinite and finite trace settings
and the relationship between planning and synthesis. 
\subsection{LTL Synthesis}

An \emph{\LTL specification} is a tuple $\synthproblem{}$ where $\varphi$ is an \LTL formula over \emph{uncontrollable} variables $\Uvars$ and 
 \emph{controllable} variables  $\Cvars$. 
A \emph{strategy} %
is a function $\strategy: (2^{\Uvars})^* \rightarrow 2^{\Cvars}$.
The \emph{infinite trace induced by 
$\mathbf{X}=\set{X_i}_{i \geq 1} \in (2^{\C{X}})^\omega$
and %
$\strategy$}
is 
\[
\inftracestrat{\strategy}{\mathbf{X}}=\left(X_1 \cup \strategy(X_1)\right) \left(X_2 \cup \strategy(X_1 X_2)\right) \ldots 
\]
The set of all infinite traces induced by $\strategy$ is denoted $\inftraces{\strategy}= \{\inftracestrat{\strategy}{\mathbf{X}}\mid 
\mathbf{X} %
\in (2^{\C{X}})^\omega
\}$.
The \emph{realizability} problem %
$\synthproblem{}$
consists in determining whether there exists a \emph{winning strategy}, i.e., 
 a strategy $\strategy$ such that $\trace \models \varphi$ for every $\trace \in \inftraces{\strategy}$. 
The \emph{synthesis} problem is to compute such a winning strategy when one exists.

\LTL synthesis can be viewed %
as a 2-player game between the environment ($\Uvars$) and the agent ($\Cvars$). 
In each turn, the environment makes a move by selecting $X_i \subseteq \Uvars$, and the agent replies by selecting $Y_i \subseteq \Cvars$. 
The aim is to find a strategy $\strategy$ for the agent that guarantees the resulting trace satisfies $\varphi$. 

\subsection{Finite LTL Synthesis}
We now recall \LTLf realizability and synthesis, where %
the specification formula is interpreted on finite traces.
An \emph{\LTLf specification} is a tuple  $\synthproblem{}$, where $\varphi$ is an \LTLf formula over \emph{uncontrollable} variables $\Uvars$ and 
 \emph{controllable} variables  $\Cvars$. 
A \emph{strategy} %
is a function $\strategy: (2^{\Uvars})^* \rightarrow 2^{\Cvars \cup \{\endtrace\}}$ such that 
for each infinite sequence $\mathbf{X}=\set{X_i}_{i \geq 1} \in (2^{\C{X}})^\omega$ 
of subsets of $\C{X}$, there is exactly one integer $n_{\strategy,\mathbf{X}}\geq 1$ with
$\endtrace \in \strategy(X_1\cdots X_{n_{\strategy,\mathbf{X}}})$. 
The induced infinite trace $\inftracestrat{\strategy}{\mathbf{X}}$ is defined as before, and the \emph{finite trace induced by 
$\mathbf{X}$ 
and %
$\strategy$}
is 
\[
\tracestrat{\strategy}{\mathbf{X}}=\left(X_1 \cup \strategy(X_1)\right) \ldots \left(X_{n_{\strategy,\mathbf{X}}} \cup \strategy(X_1\cdots X_{n_{\strategy,\mathbf{X}}})\right) 
\]
but with $\endtrace$ removed from 
$\strategy(X_1\cdots X_{n_{\strategy,\mathbf{X}}})$.
The set of all \emph{finite traces induced by $\strategy$} is denoted $\fulltraces{\strategy}= \{\tracestrat{\strategy}{\mathbf{X}}\mid \mathbf{X} \in (2^{\C{X}})^\omega\}$.
A finite trace $\trace$ is \emph{compatible with $\strategy$} if $\trace \sqsubseteq \trace'$ for some $\trace' \in \fulltraces{\strategy}$, with
$\compat{\strategy}$ (`p' for `partial') %
 the set of all such traces. 
We call $\sigma$ a winning strategy for an \LTLf\ specification $\synthproblem{}$  if $\trace \models \varphi$ for every 
$\trace \in \fulltraces{\strategy}$.
The realizability and synthesis problems for \LTLf %
are then defined %
in the same way as for \LTL. 

\paragraph{Comparison with prior formulations}
Prior work on \LTLf synthesis defined strategies as functions $\strategy: (2^{\Uvars})^* \rightarrow 2^{\Cvars}$ that do not explicitly indicate the end of the trace \cite{deg-var-ijcai15,zhu-etal-ijcai17,cam-bai-mui-mci-icaps18}. 
In these works, %
a strategy $\strategy$ is winning iff
for each $\trace \in \inftraces{\strategy}$ there exists some finite prefix $\trace' \sqsubseteq \trace$ such that $\pi' \models \varphi$.
Note that in general, multiple prefixes $\trace'$ that satisfy $\varphi$ may exist. %
We believe that it is cleaner mathematically to be precise about which trace is produced, and it will substantially simplify our technical developments. 
The two definitions give rise to the same notion of realizability, and 
existing results and algorithms for \LTLf\ synthesis transfer to our slightly different setting.

\subsection{Planning as \LTLf\ Synthesis}\label{sec:planning}
It has been observed %
that different forms of automated planning can be recast as \LTLf\ synthesis
(see e.g. \cite{deg-var-ijcai15,DIppolitoRS18,cam-bai-mui-mci-ccai18}).
We recall that planning problems are specified in terms of a set of \emph{fluents} (i.e., atomic facts whose value may change over time), a set of \emph{actions} which can change the state of the world, an \emph{action theory} whose axioms give the \emph{preconditions} and \emph{effects} of the actions (i.e., which fluents must hold for an action to be executable, and how do the fluents change as a result of performing an action), 
a description of the \emph{initial state},  and a \emph{goal}. 
In classical planning, actions are deterministic (i.e. there is a unique state resulting from performing an action in a given state), 
and the aim is to produce a sequence of actions leading from the initial state to a goal state. 
In \emph{fully observable non-deterministic} (FOND) planning, actions have non-deterministic effects, meaning that there may be multiple possible states that result from performing a given action in a given state (with the effect axioms determining which states are possible results). 
Strong solutions are \emph{policies} (i.e., functions %
that map states into actions) that guarantee eventual achievement of the goal. 

We briefly describe how FOND planning can be reduced to \LTLf\ synthesis,\footnote{Our high-level presentation 
combines elements of the reductions in \cite{deg-var-ijcai15,cam-bai-mui-mci-icaps18}. Its purpose is to illustrate the general form and components of an \LTLf\ encoding of planning
(not to provide the most efficient encoding).} 
as the reduction crucially relies on the use of environment assumptions. 
We will use the set $\C{F}$ of fluents as the uncontrollable variables, and the set of actions $\C{A}$
for the controllable variables. The high-level structure of the \LTLf\ specification formula is:
$\Phi = (\initform \land \effform) \rightarrow (\agentform \land \goalform)$
Intuitively, $\Phi$ states that under the assumption that the environment sets %
the fluents in accordance with the initial state and effect axioms 
(captured by $\initform$ and $\effform$),
the agent can choose a single action per turn ($\oneform$) in such a way that the preconditions are obeyed ($\agentform$) and the goal is achieved ($\goalform$). 
We set %
$\goalform=\ltleventually{(\gamma \land \neg\ltlnext{\top})}$, 
where $\gamma$ is a propositional formula over $\C{F}$ describing goal states. 
The formula $\oneform=\Box \psi_\mathit{one}$ with 
$\psi_\mathit{one} =\phione$
enforces that 
a single action 
is performed at each step. The formula $\agentform$ can be defined as $\Box \bigwedge_{a \in \actions} (a \rightarrow \rho_a)$, where $\rho_a$ is a propositional formula over $\fluents$ (typically, a conjunction of literals) that gives the preconditions of $a$. The formula $\initform$ will simply be the conjunction of literals over $\C{F}$ corresponding to the initial state. 
Finally, %
$\effform$ will be a conjunction of formulae of the form 
\begin{equation}
\Box \left( (\kappa \wedge a \wedge \rho_a \wedge \psi_\mathit{one}) \rightarrow \ltlweaknext{} \beta \right)
\end{equation}
where $a \in \actions$, and $\kappa$ and $\beta$ are propositional formulas over $\C{F}$. Intuitively, the latter formula 
states that if the current state verifies $\kappa$ and action $a$ is correctly performed by the agent (i.e. the preconditions are met and no other action is simultaneously performed)
then the next state must satisfy $\beta$. We discuss later why it is important to include $\rho_a \wedge \psi_\mathit{one}$.

\def\catenv{\Psi_\mathit{env}^\mathit{cat}}
\def\catagt{\Psi_\mathit{robot}^\mathit{cat}}

\subsection{Illustrative Example}\label{sec:bigexample}
We now give a concrete example of an \LTLf\ synthesis problem, which illustrates the importance of environment assumptions. 
Consider synthesizing a high-level control strategy for your Roomba-style robot vacuum cleaner. 
You want the robot to  clean the living room ($LR$) and bedroom ($BR$) when they are dirty, 
but you don't want it to vacuum a room while your cat is there
(the robot scares her).
We now describe how this problem can be formalized as \LTLf\ synthesis.

Taking inspiration from the encoding of planning, we will use
$\{clean(z), catIn(z) \mid z \in \{LR, BR\}\}$ (the fluents\footnote{We use notation reminiscent of first-order logic to enhance readability, but the variables (e.g. $clean(LR)$) are propositional.} in our scenario) as the set of uncontrollable variables,
and take the robot's actions $\{vac(BR), vac(LR)\}$ as the controllable variables.
As was the case 
for planning, it is natural to conceive of the specification as having the form of an implication $\catenv \rightarrow \catagt$, with
$\catenv$ describing the rules governing the environment's behavior and $\catagt$ the desired %
behavior of the robot. %
We define $\catagt$ as the conjunction of: %
\begin{itemize}[noitemsep,topsep=0pt]
\item for $z \in \{LR, BR\}$, the formula $\Box (vac(z) \rightarrow \rho_\mathit{vac(z)})$, 
with $\rho_\mathit{vac(z)} = \neg clean(z) \wedge \neg catIn(z)$ the precondition of $vac(z)$
(we can only vacuum dirty cat-free rooms);
\item $\Box (\neg vac(LR) \vee \neg vac(BR))$ (we cannot vacuum in two places at once);
\item $\ltleventually{(clean(LR) \wedge clean(BR))}$ (our goal: both rooms clean).  
\end{itemize}
We let $\varphi_{vac(z)}= vac(z) \wedge \rho_\mathit{vac(z)} \wedge \neg vac(z')$ (with $z'$ the other room) 
encode a correct execution of $vac(z)$, and let 
$\catenv$ be a conjunction of the following:
\begin{itemize}[noitemsep,topsep=0pt]
\item for $z \in \{LR, BR\}$: $\Box \left(clean(z) \vee \varphi_{vac(z)}\rightarrow \ltlweaknext{clean(z)}\right)$ 
(if room $z$ is currently clean, or if the robot correctly performs action $vac(z)$, then room $z$ is clean in the next state\footnote{For simplicity, we assume once a room is clean, it stays clean. })
\item for $z \in \{LR, BR\}$: $\Box \left( \neg clean(z) \wedge \neg vac(z) \rightarrow \ltlweaknext{\neg clean(z)}  \right)$ (a room can only become clean if it is vacuumed);
\item $\Box (\neg catIn(LR) \vee \neg catIn(BR))$ and $\Box(catIn(LR) \vee catIn(BR))$ (the cat is in exactly one of the rooms). 
\end{itemize}
As the reader may have noticed, while the assumptions in $\catenv$ are necessary, 
they are not sufficient to ensure realizability, as the cat may stay forever in a dirty room. 
If we further assume that 
the cat eventually leaves each of the rooms 
($\varphi_\mathit{leaves} =\ltleventually{\neg catIn(BR)} \wedge \ltleventually{\neg catIn(LR)}  $),
there is an obvious solution: 
vacuum a cat-free room, and then simply wait until the other room is cat-free and then vacuum it.
However, rather unexpectedly, adding $\varphi_\mathit{leaves}$ to $\catenv$ makes the specification $\catenv \rightarrow \catagt$
realizable in a trivial and unintended way: by ending execution in the first move, $\neg \catenv$ trivially holds in the resulting length-1 trace~$\pi$. Indeed, 
there are three possibilities: (i)  $\pi \models catIn(BR)$ (so $\pi \not \models \ltleventually{\neg catIn(BR)}$),
(ii)  $\pi \models catIn(LR)$ (so $\pi \not \models \ltleventually{\neg catIn(LR)}$), or
(iii) $\pi \models \neg catIn(LR) \wedge \neg catIn(BR)$ (so $\pi \not \models  \left (\Box(catIn(LR) \vee catIn(BR)) \right)$). 
Clearly, 
this length-1 strategy 
is not the strategy that we wanted to synthesize.
In Section \ref{sec:constrained}, we propose a new framework for handling environment assumptions
which avoids the generation of such trivial strategies and makes it possible to find the desired strategies.

\section{Constrained \LTLf\ Synthesis}\label{sec:constrained}

To the aim of properly handling environment assumptions, 
we introduce a generalization of \LTLf\ synthesis, %
in which the assumptions are separated from the rest of the specification formula
and given a different interpretation. Essentially, the idea is that the environment is allowed to 
satisfy the assumption over the whole infinite trace, rather than on the finite prefix chosen by the agent.  
This can be accomplished using \LTL\ semantics for the environment assumption, 
but keeping \LTLf\ semantics for the formula describing the objective. %

Formally, a \emph{constrained \LTLf specification} is a tuple $\constrainedproblem{}$, 
where $\Uvars$ and $\Cvars$ are the uncontrollable and controllable variables,  
$\varphi$ is an \LTLf\ formula over $\Uvars \cup \Cvars$, and  $\assume$ is an \LTL formula over $\Uvars \cup \Cvars$.  
Here $\varphi$ describes the desired agent behavior when the environment behaves so as to satisfy $\assume$. 
We will henceforth call $\varphi$ the \emph{objective}, and will refer to $\assume$ as the (environment) \emph{assumption} or \emph{constraint} 
(as it acts to constrain the allowed environment behaviors).

A \emph{strategy for \constrainedproblem{}}
is a function $\strategy: (2^{\Uvars})^* \rightarrow 2^{\Cvars \cup \{\endtrace\}}$ such that 
for each infinite sequence 
$\mathbf{X}=\set{X_i}_{i \geq 1} \in (2^{\C{X}})^\omega$ 
of subsets of $\C{X}$, there is \emph{at most} one integer $n_{\strategy,\mathbf{X}}\geq 1$ with
$\endtrace \in \strategy(X_1\cdots X_{n_{\strategy,\mathbf{X}}})$. If none exists, we write $n_{\strategy,\mathbf{X}} = \infty$.
To account for traces that do not contain $\endtrace$, we redefine 
$\fulltraces{\strategy}$ as follows: 
$\{\tracestrat{\strategy}{\mathbf{X}}\mid \mathbf{X} \in (2^{\C{X}})^\omega \text{ and } n_{\strategy,\mathbf{X}}<\infty\}$.
A strategy $\strategy$ is an \emph{$\assume$-strategy} if 
for every $\mathbf{X} \in (2^{\C{X}})^\omega$, either 
$n_{\strategy,\mathbf{X}} < \infty$ or $\inftracestrat{\sigma}{\mathbf{X}} \not \models \assume$,
i.e. $\strategy$ terminates on every trace that satisfies~$\assume$.  
A \emph{winning strategy} (w.r.t. \constrainedproblem{})
is an $\assume$-strategy such that
$\trace \models \varphi$ for every $\trace \in \fulltraces{\strategy}$.
In other words, winning strategies are those that guarantee the satisfaction of the objective $\varphi$ 
under the assumption that the environment behaves in a way that constraint $\assume$ is satisfied. 
The realizability and synthesis problems for constrained \LTLf specifications are defined as before, using this notion of winning strategy. 

Because the constraints are interpreted using infinite \LTL\ semantics,
we are now able to correctly handle 
liveness constraints %
($\ltleventually{\psi}$) and 
fairness constraints as studied in \LTL synthesis ($\ltlalways{\ltleventually{\psi}}$) and FOND planning ($\ltlalways{\ltleventually{\psi_1}} \rightarrow \ltlalways{\ltleventually{\psi_2}} $) \cite{DIppolitoRS18}.

\begin{example}
Returning to our earlier example, consider the constrained synthesis problem with assumption $\catenv$ (which includes $\varphi_\mathit{leaves}$) and objective $\catagt$. 
The obvious strategy 
(vacuum dirty rooms as soon as they are cat-free) 
gives rise to a winning strategy, in which we output $\endtrace$ if we manage to clean both rooms, and otherwise, produce an infinite trace without $\endtrace$ in which $\catenv$ is not true.
Trivial strategies that terminate immediately will \emph{not} be winning strategies, 
as there will be infinite traces that satisfy the constraint but where the length-1 finite 
trace falsifies the objective. 
\end{example}

We remark that %
if we are not careful about how we write the constraint $\assume$,
we may unintentionally allow the agent to block the environment from fulfilling $\assume$.

\begin{remark}
Suppose that instead of using Equation 1 to encode the effects of actions, we employ the simpler
 $\Box \left( (\kappa \wedge a) \rightarrow \ltlweaknext{} \beta \right)$. 
While intuitive, this alternative formulation does not properly encode %
FOND planning, as the specification may be realized in an unintended way: 
by performing multiple actions with conflicting effects, or a single action whose precondition is not satisfied, the agent can force the environment to satisfy a contradictory set of formulae $\beta$ in the next state, causing the assumption to be violated. 
\end{remark}

\acite{ChatterjeeHJ08} discuss this phenomenon in the context of \LTL synthesis, and suggest
that a reasonable environment constraint is one which is \emph{realizable for the environment}. 
We note that the constraints we considered in Section \ref{sec:synth-env-ass} all satisfy this property.

\subsubsection{Correspondence with Finite LTL Synthesis}
We begin by observing that (plain) \LTLf\ synthesis is a special case of constrained \LTLf\ synthesis in which one uses the trivial constraint $\top$ for the environment assumption:

\begin{theorem}\label{thm:reduction}
Every winning strategy $\strategy$ for the \LTLf\ specification $\la \C{X}, \C{Y}, \varphi \ra$ is a winning strategy for the constrained \LTLf\ specification $\la \C{X}, \C{Y}, \top, \varphi \ra$, and vice-versa. In particular, 
$\la \C{X}, \C{Y}, \varphi \ra$ is realizable iff $\la \C{X}, \C{Y}, \top, \varphi \ra$ is realizable. 
\end{theorem}

A natural question is whether a reduction in the other direction exists. Indeed, it is well-known that in the infinite setting, 
assume-guarantee \LTL\ synthesis\footnote{Here we refer to assume-guarantee synthesis as considered in \cite{ChatterjeeHJ08,AlmagorKRV17}, where given a pair $(\assume, \varphi)$, the aim is
to construct a strategy such that every induced infinite trace either violates $\assume$ or satisfies $\varphi$. This is different from the assume-guarantee synthesis of \cite{ChatterjeeH07}, in which
N agents each have their own goals, and the objective is for each agent to satisfy its own goals.} 
with an assumption $\assume$ and objective $\varphi$ corresponds to classical \LTL\  synthesis w.r.t.\ $\assume \rightarrow \varphi$ (that is, the two synthesis problems have precisely the same winning strategies). 
The following negative result shows that a simple reduction via implication does not work in the finite trace setting:

\begin{theorem}\label{thm:implication-fails}
There exists an unrealizable constrained  \LTLf specification $\C{S}=\la \C{X}, \C{Y},
 \assume, \varphi \ra$
such that the \LTLf\ specification $\C{S}_\rightarrow = \la 
\C{X}, \C{Y},
 \assume \rightarrow \varphi \ra$ is realizable. 
\end{theorem}
\begin{proof}
Consider the constrained \LTLf specification $\C{S}=\la \set{x,x'}, \set{y}, \assume, \varphi\ra$ with $\assume=\neg x \wedge \ltleventually{x}$
and $\varphi= \ltleventually{(x' \!\wedge y)}$. We claim that $\C{S}$ is unrealizable. 
Indeed, take any $\mathbf{X}=X_1 X_2 \ldots$ such that $x \not \in X_1$, $x \in X_2$, and $x' \not \in X_i$ for all $i \geq 1$. 
Then no matter which strategy $\sigma$ is used, the infinite trace $\inftracestrat{\sigma}{\mathbf{X}}$ 
will satisfy $\assume$, and the induced finite trace $\tracestrat{\sigma}{\mathbf{X}}$, if it exists, will falsify $\varphi$ (as $x'$ never holds). 

Next consider %
$\C{S}_\rightarrow = \la \set{x}, \set{y}, \assume \rightarrow \varphi \ra$, and observe 
that $\assume \rightarrow \varphi \equiv x \vee (\Box \neg x) \vee \ltleventually{(x' \!\wedge y)}$. 
A simple winning strategy exists: output $\endtrace$ in the first time step. Indeed, %
every induced trace has length 1 and hence trivially satisfies  $x \vee \Box \neg x$. %
\end{proof}

With the next theorem, we observe a more fundamental difficulty in reducing constrained \LTLf\ synthesis problems to standard \LTLf\ synthesis: 
winning strategies for constrained problems may need an \emph{unbounded} number of time steps to realize the specification, 
a phenomenon that does not occur in standard \LTLf synthesis. 

\begin{theorem}\label{thm:no_reduction}
An \LTLf\ specification is realizable iff it admits a \emph{bounded} winning strategy, i.e.\ a strategy for which there exists %
$B > 0$
such that every induced finite trace has length at most $B$. %
There exist realizable constrained \LTLf\ specifications that do not possess any bounded winning strategy. 
\end{theorem}
\begin{proof}
A straightforward examination of the \LTLf\ synthesis algorithm\footnote{The algorithm can be easily modified to output $\endtrace$ once $\varphi$ has been satisfied to match our definition of strategy.  } in \cite{deg-var-ijcai15} shows that 
when %
$\la \Uvars, \Cvars, \varphi\ra$ is realizable, the produced strategy guarantees achievement of $\varphi$ in 
a number of time steps bounded by the number of states in a DFA for $\varphi$. 

For the second point, 
consider the constrained \LTLf specification $\C{S}=\la \set{x}, \set{y}, \ltleventually{x}, \ltluntil{\neg y\ }{(x \land y)} \ra$. 
Observe that %
$\C{S}$ is realizable, as it suffices to output $\neg y$ until the first $x$ is read, then output $\{y,\endtrace\}$.
Assume for a contradiction that there is a winning strategy $\sigma$ for $\C{S}$ and constant $B > 0$ such that 
 $n_{\strategy,\mathbf{X}}\leq B$ for every $\mathbf{X} \in \C{X}^\omega$. 
Define $\mathbf{X}^B = X_1^B X_2^B \ldots $ as follows: $X_i^B=\{x\}$ if $i=B+1$ and %
$X_i^B=\emptyset$ otherwise. 
The induced trace $\pi = \tracestrat{\sigma}{\mathbf{X}^B}$ has length at most $B$ and hence does not contain $x$.
It follows that $\pi \not \models \varphi$, contradicting our assumption
that $\sigma$ is a winning strategy. 
\end{proof}

While the implication-based approach does not work in general, 
we show that it can be made to work for environment assumptions that belong to the safe fragment:

\begin{theorem}\label{thm:safe_reduction}
When $\assume$ is a safe formula, 
the constrained \LTLf\ specification $\C{S}=\constrainedproblem{}$ is realizable iff the \LTLf specification $\C{S}'=\langle \Uvars, \Cvars, \assume' \rightarrow \varphi \rangle$ is realizable,
where $\assume'$ is %
obtained from $\assume$ by replacing every occurrence of $\ltlnext{\psi}$ by $\ltlweaknext{\psi}$. 
\end{theorem}
\begin{proof}[Proof sketch]
Let $\strategy'$  be a winning strategy for $%
\langle \Uvars, \Cvars, \assume' \rightarrow \varphi \rangle$,  
with $\assume$ a safe formula. To define a winning strategy $\sigma$ for $%
\constrainedproblem{}$, we set $\sigma(X_1\cdots X_n)$ equal to
\begin{itemize}[noitemsep,topsep=0pt]
\item $\sigma(X_1 \cdots X_n) \setminus \{\endtrace\}$, when $\endtrace \in \sigma(X_1 \cdots X_n)$ and \newline $(X_1 \cup \strategy(X_1)) \ldots (X_n \cup \strategy(X_1\cdots X_n)) \not \models \assume'$; 
\item $\sigma(X_1 \cdots X_n)$, otherwise.
\end{itemize}
For the other direction, given a winning strategy $\strategy$ for $\constrainedproblem{}$, %
we can define a winning strategy $\sigma'$ for %
$\langle \Uvars, \Cvars, \assume' \rightarrow \varphi \rangle$ 
by setting $\sigma'(X_1\cdots X_n)$ equal to
\begin{itemize}
\item $\strategy(X_1\cdots X_n) \cup \set{\endtrace}$, if $(X_1 \cup \strategy(X_1)) \ldots (X_n \cup \strategy(X_1\cdots X_n))$ is a bad prefix for $\assume$,
and $\endtrace \not \in \sigma'(X_1\cdots X_k)$ for $k < n$;
\item $ \strategy(X_1\cdots X_n) \setminus \set{\endtrace}$, if 
$\endtrace \in \sigma'(X_1\cdots X_k)$ for some $k < n$;
\item $\sigma'(X_1\cdots X_n)= \strategy(X_1\cdots X_n)$, otherwise.\qedhere
\end{itemize} 
\end{proof}

The following example shows that it is essential in the preceding theorem to use $\assume' \rightarrow \varphi$ rather than $\assume \rightarrow \varphi$:

\begin{example}
If we let $\assume=\Box(\neg x \vee \ltlnext{x})$  %
and $\varphi = \neg x \wedge y$, then the constrained specification $\constrainedproblem{}$
is not realizable (as the environment can output $x$ in the first step), but the \LTLf\ specification $\langle \Uvars, \Cvars, \assume \rightarrow \varphi \rangle$ is realizable with a strategy that outputs $\{y, \endtrace\}$ in the first step. 
Indeed, if the environment outputs $x$, then $\neg \assume \equiv \ltleventually{(x \wedge \neg \ltlnext{x})}$ holds in the induced length-1 trace;
if we have $\neg x$ instead, then $\varphi$ holds. %
\end{example}
 
Note however that the negative result in the general case (Theorem \ref{thm:implication-fails}) continues to hold if $\assume' \rightarrow \varphi$ is used instead of $\assume \rightarrow \varphi$, since the formulas in that proof do not involve $\ltlnext{}$. 

Another interesting observation is the environment assumptions $\initform$ and $\effform$
used to encode the initial state and action effects in planning are safe formulas. 
This explains why these constraints can be properly encoded in \LTLf\ using implication 
and $\ltlweaknext{}$ rather than $\ltlnext{}$. We note that if we encode planning 
using constrained \LTLf\ synthesis, then we can use $\ltlnext{}$ in the effect axioms, 
which is arguably more natural.

\subsubsection{Reduction to LTL Synthesis}
\def\translate{\tau}
Every \LTLf formula $\varphi$ over %
$\C{P}$ can be polynomially transformed into an %
\LTL formula $\varphi_{\mathsf{inf}}$ over %
$\C{P} \cup \set{\mathsf{alive}}$ such that %
$\pi \models \varphi_{\mathsf{inf}}$ iff $\pi' \models \varphi$ for some finite prefix $\trace' \sqsubseteq \trace$ \cite{DBLP:conf/ijcai/GiacomoV13}.
Intuitively, %
$\mathsf{alive}$ holds for the duration of the (simulated) finite trace. 
Formally, 
$\varphi_{\mathsf{inf}} \defeq \translate(\varphi) \land \mathsf{alive} \land \left( \ltluntil{\mathsf{alive}}{(\ltlalways{\neg \mathsf{alive}})} \right)$, 
where:
\begin{center}
$  \translate(p) = p %
 \quad
\translate(\neg \varphi) = \neg \translate(\varphi) \quad
    \translate(\varphi_1 \land \varphi_2) =\translate(\varphi_1) \land \translate(\varphi_2)$\\
    $    \translate(\ltlnext{\varphi}) = \ltlnext{(\mathsf{alive} \land \translate(\varphi))} \hfill
    \translate(\ltluntil{\varphi_1}{\varphi_2}) = \ltluntil{\tau(\varphi_1)}{(\mathsf{alive} \land \translate(\varphi_2))} $
\end{center}
We extend this transformation as follows: 
\begin{align*}
    \psi_{\endtrace} &\defeq \ltlalways{(\endtrace \leftrightarrow \mathsf{alive} \land \ltlnext{\neg \mathsf{alive}})} \land \ltlalways{(\endtrace \rightarrow \ltlnext{\ltlalways{ \neg \endtrace}})} \\
    \psi_{\assume, \varphi} &\defeq \psi_{\endtrace} \land \left ((\assume \vee \ltleventually{\endtrace}) \rightarrow \varphi_{\mathsf{inf}}\right)
\end{align*}
Here $\psi_{\endtrace}$ forces the agent to trigger variable $\endtrace$ when the end of the trace is simulated and also ensures that $\endtrace$ occurs at most once. 
Formula $\psi_{\assume, \varphi}$ ensures that $\varphi_{\mathsf{inf}}$ is satisfied -- i.e., a finite trace that satisfies $\varphi$ and ends is simulated -- when either the environment assumption $\assume$ holds or $\endtrace$ occurs.

\begin{theorem}\label{reduction-infinite}
The constrained \LTLf\ specification $\C{S}=\constrainedproblem{}$ is realizable iff \LTL\ specification $\C{S}^\infty=\langle \Uvars, \Cvars \cup \set{\mathsf{alive},\endtrace}, \psi_{\assume, \varphi} \rangle$ is realizable. Moreover, 
for every winning strategy $\sigma$ for $\C{S}^\infty$, 
the strategy $\sigma'$ defined by $\strategy'(X_1\cdots X_n) \defeq \strategy(X_1 \cdots X_n) \setminus \set{\mathsf{alive}}$ is a winning strategy for $\C{S}$. 

\end{theorem}

\section{Algorithms for Constrained \LTLf\ Synthesis}
\label{sec:alg-cltlf-synth}
\LTL and \LTLf realizability are both 2EXP-complete  \cite{pnu-ros-popl89,deg-var-ijcai15}, and we can show the same holds for constrained \LTLf\ problems. The upper bound exploits the reduction to \LTL\ (Theorem \ref{reduction-infinite}), and the lower bound is inherited from (plain) \LTLf synthesis, which is a special case of constrained \LTLf\ synthesis (Theorem \ref{thm:reduction}). 

\begin{theorem}\label{thm:complexity}
Constrained \LTLf realizability (resp.\ synthesis) is 2EXP-complete (resp.\ in 2EXP).  
\end{theorem}

It follows from Theorem \ref{thm:complexity} that the reduction to infinite \LTL\ realizability and synthesis
yields worst-case optimal algorithms. 
However, we argue that the reduction to LTL  %
does not provide a practical approach. 
Indeed, while \LTL\  and \LTLf\ synthesis share the same worst-case complexity, 
recent experiments have shown 
that \LTLf\ is much easier to handle in practice \cite{zhu-etal-ijcai17}. Indeed, 
state-of-the-art approaches to \LTL synthesis rely on first translating the \LTL\ formula into a suitable infinite-word automata,
then solving a two-player game on the resulting automaton. The computational bottleneck is the complex 
transformations of infinite-word automata, for which no efficient implementations exist. %
Recent approaches to \LTLf\ synthesis also adopt an automata-game approach, but 
\LTLf\ formulae require only finite-word automata (NFAs and DFAs), which can be manipulated more efficiently. 

The preceding considerations motivate us to explore an alternative approach to constrained \LTLf synthesis,
which involves a reduction to \DBA games. Importantly, the \DBA\ can be straightforwardly constructed from DFAs 
for the constraint and objective formulae, allowing us to sidestep the difficulties of manipulating infinite-word automata.

\subsection{\DBA for Constrained Specifications} \label{dba-construction}
For the rest of this section, we assume $\assume = \assume_s \land \assume_c$, where $\assume_s$ is a safe formula and  $\assume_c$  is a co-safe formula\footnote{
If we want to have \emph{only} a safe 
(resp.\ co-safe) constraint, it suffices to use a trivial constraint $\assume_c=\top$ (resp.\ $\assume_s = \bot  \ltlrelease{}{}(p \vee \neg p)$). }, both defined over 
$\C{X}\cup \C{Y}$. Safe and co-safe formulae are well-known \LTL\ fragments \cite{DBLP:journals/fmsd/KupfermanV01} of  proven utility. Safe formulas are prevalent in \LTL\ %
specifications and a key part of the encoding of planning as \LTLf\ synthesis (see Section \ref{sec:planning}); the usefulness of co-safe formulas can be seen from our example (Section \ref{sec:bigexample}) and their adoption in work on robot planning (see e.g.\ \cite{DBLP:conf/aaai/LahijanianAFKV15}).

Our aim %
is to construct a \DBA that accepts infinite traces $\pi$ over $2^{\C{X}\cup \C{Y} \cup \set{\endtrace}}$
such that either (i) $\pi$ contains a single occurrence of $\endtrace$ which induces a finite prefix $\pi'$ with $\pi' \models \varphi$,
or (ii) $\pi$ doesn't contain $\endtrace$ and $\pi \not \models \assume_s \wedge \assume_c$. 
Such a \DBA $\constraintdba$ %
can be defined by combining three DFAs:
 $\C{A}_s = \langle \alphPV, Q_s, \delta_s, (q_0)_s, \QF{s} \rangle$ accepts the bad prefixes of $\assume_s$;
$\C{A}_c = \langle \alphPV, Q_c, \delta_c, (q_0)_c, \QF{c} \rangle$ accepts the good prefixes of $\assume_c$;
and $\C{A}_g = \langle \alphPV, Q_g, \delta_g, (q_0)_g, \QF{g} \rangle$ accepts models of $\varphi$. 
Recall that these DFAs can be built in double-exponential time.

Formally, we let $\constraintdba = \langle 2^{\C{P}\cup \{\endtrace\}}, Q, \delta, q_0, \QF{}\rangle$,  
where $Q$, $q_0$, and $\QF{}$ are defined as follows: %
\begin{itemize}[noitemsep,topsep=0pt]
\item $Q \defeq \left ((Q_s\cup\{q_{\mathsf{bad}}\}) \times (Q_c\cup\{q_{\mathsf{good}}\}) \times Q_g) \right)\cup \{q_\top, q_\bot\} $
\item $q_0 \defeq ((q_0)_{s}, (q_0)_{c}, (q_0)_{g})$ 
\item $\QF{} \defeq \{(q_s,q_c,q_g) \in Q \mid 
q_s = q_\mathsf{bad} \text{ or } q_c\neq q_\mathsf{good} \} \cup  \{q_\top\}$
\end{itemize}
For `regular' symbols $\alphsym \in 2^{\C{P}}$ (i.e., $\endtrace \not \in \alphsym$), 
we set $\delta((q_s,q_c,q_g), \alphsym)=(\delta_s^*(q_s,\alphsym), \delta_c^*(q_c,\alphsym), \delta_g(q_g, \alphsym))$ where:
{
\small
\begin{align*}
\delta_s^*(q_s,\alphsym) &= 
\begin{cases}
q_\mathsf{bad}, & \text{if } 
q_s = q_\mathsf{bad} \text{ or } \delta_s(q_s,\alphsym) \in \QF{s}\\
\delta_s(q_s,\alphsym),& \text{otherwise}
\end{cases}\\
\delta_c^*(q_c,\alphsym) &= 
\begin{cases}
q_\mathsf{good}, & \text{if } 
q_c = q_\mathsf{good} \text{ or } \delta_c(q_c,\alphsym) \in \QF{c}\\
\delta_c(q_c,\alphsym),& \text{otherwise}
\end{cases}
\end{align*}}
For $\alphsym$ with $\endtrace \in \alphsym$,
we set $\delta((q_s,q_c,q_g), \alphsym)=q_\top$ if 
$\delta_g(q_g,\alphsym) \in \QF{g}$,
and $\delta((q_s,q_c,q_g), \alphsym)=q_\bot$ in all other cases. %
Accepting state $q_\top$ is quasi-absorbing: $\delta(q_\top, \alphsym) = q_\top$ when $\endtrace \not\in \alphsym$, and $\delta(q_\top, \alphsym) = q_\bot$ otherwise. This forces winning strategies to output variable $\endtrace$ at most once.
Finally, $q_\bot$ is an absorbing state: $\delta(q_\bot, \alphsym) = q_\bot$ for every $\alphsym \in 2^{\C{P} \cup \{\endtrace\}}$.

\begin{theorem}\label{dba-thm}
The \DBA $\constraintdba$ accepts infinite traces $\pi$ such that either: (i) 
$\pi(1)\cdots\pi(n) \models \varphi$ and $\endtrace$ occurs only in $\pi(n)$,
or (ii) $\pi \not\models \assume_s \wedge \assume_c$ and $\endtrace$ does not occur in $\pi$. 
$\constraintdba$ can be constructed in double-exponential time in 
$|\assume_s| + |\assume_c| + |\varphi|$.
\end{theorem}

\subsection{\DBA Games}\label{ssec:dba-games}
Once a %
specification has been converted into a \DBA, 
realizability and synthesis can be reduced to \DBA\ games. We briefly recall next
the definition of such games and how winning strategies can be computed.

A \emph{\DBA (or B\"{u}chi) game} (see e.g.\ \cite{ChatterjeeHenzingerPiterman06_AlgorithmsForBuchiGames}) 
is a two-player game given by a tuple $\langle \C{X}, \C{Y}, \C{A} \rangle$, 
where $\C{X}$ and $\C{Y}$ are disjoint finite sets of variables and $\C{A}$ is a \DBA with alphabet $2^{\C{X} \cup \C{Y}}$.
A \emph{play} %
is an infinite sequence of rounds, where in each round, Player I
selects %
$X_i \subseteq \C{X}$, then Player II %
selects %
$Y_i \subseteq \C{Y}$. 
A play is winning if it yields a word  $(X_1 \cup Y_1) (X_2 \cup Y_2) \ldots $ that belongs to $\lang(\C{A})$.
A game is \emph{winning} if there exists a strategy  $\strategy: (2^{\Uvars})^* \rightarrow 2^{\Cvars}$ 
such that for every infinite sequence $X_1 X_2 \ldots \in \C{X}^\omega$, 
the word $(X_1 \cup \strategy(X_1)) (X_2 \cup \strategy(X_1 X_2)) \ldots$ obtained by following $\strategy$ belongs to $ \lang(\C{A})$.
In this case, we call $\sigma$ a \emph{winning strategy}.

Existence of a winning strategy for a DBA game $\C{G} = \langle \C{X}, \C{Y}, \C{A} \rangle$
based upon $\C{A} = \langle 2^{\C{X}\cup \C{Y}}, Q, \delta, q_0, \QF{}\rangle$ 
can be determined by computing the winning region of $\C{G}$. This is done in two steps.
First, we compute the set $RA(\C{G})$ of \emph{recurring accepting} states, 
i.e. those $q \in \QF{}$ such that Player II has a strategy from state $q$ to revisit $\QF{}$ %
infinitely often.
Next, we define the \emph{winning region} $Win(\C{G})$ of $\C{G}$ as those states in $q \in Q$ for which 
Player II has a strategy for reaching a state in $RA(\C{G})$.
The sets 
$RA(\C{G})$ and $Win(\C{G})$ can be computed in polynomial time by utilizing the
\emph{controllable predecessor} operator: 
$\cpredecessor(S) = \{q \in Q \mid \forall X \subseteq \C{X}\, \exists Y \subseteq \C{Y}:  \delta(q,X\cup Y) \in S\}$.
We set $\win^{0}(S) = S$ and 
$\win^{i+1}(S) = \win^{i}(S) \cup \cpredecessor(\win^{i}(S))$. 
Intuitively, $\win^i(S)$ contains those states from 
which Player II has a strategy for reaching (or returning to) $S$ in at most $i$ rounds.  
The limit $\lim_i \win^i(S)$ exists because $\win^{i}(S) \subseteq \win^{i+1}(S)$, and convergence is achieved in a finite number of iterations bounded by $|Q|$.
To compute $RA(\C{G})$, 
we set $S_{\!1} = \QF{}$ and let $S_{\!k+1}=S_{\!k} \cap \lim_i \win^i(S_k)$. The set $S_{\!k}$ contains those accepting states from which Player II has a strategy for visiting $S_{\!k}$ no less than $k$ times.
The limit $\lim_k S_k$ exists because $S_{\!k}  \subseteq S_{\!k+1}$, and convergence is achieved in a finite number of iterations bounded by $|\QF{}|$.
$RA(\C{G})$ is the finite limit of $S_{\!k}$, and
the set $Win(\C{G})$ is then the finite limit of $\win^i(RA(\C{G}))$.
It is easy to see that $Win(\C{G})$ can be computed in polynomial time w.r.t.\ the size of the DBA $\C{A}$. 
The following well-known result shows how we can use 
$Win(\C{G})$ to decide if $\C{G}$ is winning.

\begin{theorem}\label{thm:winning}
$\C{G}$ is winning iff $q_0 \in Win(\C{G})$.
\end{theorem}

We sketch the proof of the right-to-left implication here, since it will be needed for later results. 
We suppose that $q_0 \in Win(\C{G})$ and show how to construct a transducer
$\C{T}_\C{G}$ that implements a winning strategy. %
Intuitively, the transducer's 
output function ensures that the transducer stays within $Win(\C{G})$, always reducing the `distance' to $RA(\C{G})$.
More precisely, we can define $\C{T}_\C{G}$ as $\langle 2^\C{X}, 2^\C{Y}, Q, \delta', \omega, q_0 \rangle$, where: 
the set of states $Q$ and initial state $q_0$ are the same as for the DBA $\C{A}$, 
and the transition function $\delta'$ mirrors the transition function $\delta$ of $\C{A}$:
$\delta'(q, X) =  \delta(q,X \cup \omega(q,X))$.
We define the output function $\omega$ as follows:
\begin{itemize}[noitemsep,topsep=0pt]
\item Case 1: there exists $Y^*$ such that $\delta(q, X \cup Y^*) \in  Win(\C{G})$. In this case, 
we let $\omega(q,X)$ be any\footnote{Several $Y$ may satisfy the conditions, and choosing \emph{any} such $Y$ yields a suitable transducer. Alternatively, one can use nondeterministic transducers (called \emph{strategy generators} in \cite{deg-var-ijcai15}) to encode a family of deterministic transducers. }  $Y \in 2^\C{Y}$ such that 
(a) $\delta(q, X \cup Y) \in  \win^{i+1}(RA(\C{G}))$, and
(b) there is no $Y'$ with $\delta(q, X \cup Y') \in  \win^i(RA(\C{G}))$.
\item Case 2: no such $Y^*$ exists. We let $\omega(q,X)$ be any $Y \in 2^\C{Y}$.
\end{itemize}
According to this definition, after reading $X$, the transducer $\C{T}_\C{G}$ chooses an output symbol $Y$ that allows the underlying automaton $\C{A}$ 
to transition from the current state via $X \cup Y$ to a state in $Win(\C{G})$ (if some such symbol exists). Moreover, 
among the immediately reachable winning states, preference is given to those that are closest to $RA(\C{G})$, i.e. those belonging
to $\win^{i}(RA(\C{G}))$ for the minimal value $i$. 

\subsection{Constrained \LTLf\ Synthesis via \DBA Games} 
Given a constrained \LTLf synthesis specification \constrainedproblem{} with $\assume = \assume_s \land \assume_c$, we proceed as follows:
\begin{enumerate}[noitemsep,topsep=0pt]
\item Construct the \DBA\ game $\C{G}^{\assume_s, \assume_c}_{\varphi} = \langle \C{X}, \C{Y} \cup \{\endtrace\}, \constraintdba\rangle$.
\item Determine whether $\C{G}^{\assume_s, \assume_c}_{\varphi} $ is winning: build $Win(\C{G}^{\assume_s, \assume_c}_{\varphi} )$ and check whether $((q_0)_{s}, (q_0)_{c}, (q_0)_{g}) \in Win(\C{G}^{\assume_s, \assume_c}_{\varphi} )$. 
\item If $\C{G}^{\assume_s, \assume_c}_{\varphi} $ is not winning, return `unrealizable'.  
\item Otherwise, compute a winning strategy for $\C{G}^{\assume_s, \assume_c}_{\varphi}$ using the transducer from Section \ref{ssec:dba-games}. 
\end{enumerate}
Using Theorems \ref{dba-thm} and \ref{thm:winning}, we can show that this method is correct and yields optimal complexity:

\begin{theorem}\label{thm:correct-constrained}
Consider a constrained \LTLf specification $\C{S}=$\agproblem{\assume_s \wedge \assume_c}{\varphi} where $\assume_s$ (resp.\ $\assume_c$) is a safe (resp.\ co-safe) formula. 
Then:
\begin{itemize}[noitemsep,topsep=0pt]
\item  $\C{S}$ is realizable iff the \DBA game $\C{G}^{\assume_s, \assume_c}_{\varphi}$ is winning;
\item Every winning strategy for 
$\C{G}^{\assume_s, \assume_c}_{\varphi}$ is a winning strategy for $\C{S}$, and vice-versa;
\item Deciding whether $\C{G}^{\assume_s, \assume_c}_{\varphi}$ is winning, and constructing a winning strategy when one exists,
can be done in 2EXP. 
\end{itemize}
\end{theorem}

\section{Synthesis of High-Quality Strategies}
\label{sec:high-quality-synth}

This section explores the use of a quantitative specification language to compare strategies based upon how well they satisfy the specification. We adopt the \LTLfF language from \cite{AlmagorBK16} and propose a new more refined way of defining optimal strategies.

\subsection{The Temporal Logic \LTLfF}
We recall here the language \LTLfF\ proposed by \acite{AlmagorBK16}. The basic idea is that instead of a formula being either totally satisfied or totally violated by a trace,  a value between 0 and 1 will indicate its \emph{degree of satisfaction}. In order to allow for different ways of aggregating formulae, the basic \LTL\ syntax is augmented with a set $\Fmc \subseteq \{f: [0,1]^k\rightarrow [0,1] \mid k \in \mathbb{N} \}$ of functions, with the choice of which functions to include in $\Fmc$ being determined by the application at hand. 

Formally, the set of \LTLfF\ formulae is obtained by adding $f(\varphi_1, \ldots, \varphi_k)$ to the grammar for $\varphi$, for every $f \in \Fmc$. 
We assign a \emph{satisfaction value} to every \LTLfF\ formula, finite trace $\trace$, and time step $1 \leq i \leq |\trace|$, as follows\footnote{We omit $\vee$ and $\ltlrelease{}{}$, as they can be defined using $\neg$, $\wedge$, and $\ltluntil{}{}$.}: 
{\small
\begin{align*}
\eval{\trace}{\top}{i}&= 1 \quad \eval{\trace}{\bot}{i}=0 \quad
\eval{\trace}{p}{i} = \begin{cases} 1 & \mbox{if } p \in \trace(i)\\  0 & \mbox{otherwise}  \end{cases}\\
\eval{\trace}{\neg \varphi}{i} &= 1 - \eval{\trace}{ \varphi}{i} \\
\eval{\trace}{\varphi_1 \wedge \varphi_2}{i} &= \min\{\eval{\trace}{ \varphi_1}{i}, \eval{\trace}{ \varphi_2}{i}\}\\
\eval{\trace}{\ltlnext{} \varphi}{i}&= \eval{\trace}{\varphi}{i+1}\\
\eval{\trace}{f(\varphi_1, \ldots, \varphi_k)}{i}&= f(\eval{\trace}{\varphi_1}{i}, \ldots, \eval{\trace}{\varphi_k}{i})\\
\eval{\trace}{\varphi_1 \ltluntil{}{} \varphi_2}{i}&= \max_{i \leq i' \leq |\trace|} \left\{ \min \left\{ \eval{\trace}{\varphi_2}{i'}, \min_{i \leq j < i'} \{  \eval{\trace}{\varphi_1}{j}        \}\right\} \right \} 
\end{align*}}
The (satisfaction) value of $\varphi$ on $\trace$, written $\eval{\trace}{\varphi}{}$, is $\eval{\trace}{\varphi}{1}$.

We define $\pvals(\varphi) \subseteq [0,1]$ as the set of values $\eval{\trace}{\varphi}{i}$, ranging over all traces 
$\trace$ and steps $1 \leq i \leq |\trace|$. 
The following proposition, proven by \cite{AlmagorBK16}, shows that an \LTLfF\ formula 
can take on only exponentially many different values.

\begin{proposition}\label{exp-values}
For every \LTLfF\ formula $\varphi$, $|V(\varphi)| \leq  2^{|\varphi|}$. 
\end{proposition}

The functions $f$ allow us to capture a diversity of methods for combining 
a set of potentially competing objectives (including classical preference 
aggregation methods like weighted sums and lexicographic ordering). 

\begin{example}
For illustration purposes, 
consider two variants of our robot vacuum example,
with specifications $\varphi_1$ and $\varphi_2$, 
in which the goal $\ltleventually{} ( clean(LR) \land clean(BR))$ is replaced by $\ltleventually{} ( clean(LR) )$ and $\ltleventually{} ( clean(BR))$, respectively.
We can include in $\C{F}$ a binary weighted sum operator $sum_{0.3,0.7}$, where %
the satisfaction value of $sum_{0.3,0.7}(\varphi_1, \varphi_2)$ on trace $\pi$ is 0.3 if $\pi \models \varphi_1 \land \neg \varphi_2$, 0.7 if  $\pi \models \neg \varphi_1 \land \varphi_2$, 1 if  $\pi \models \varphi_1 \land \varphi_2$, and zero otherwise. We can thus express that we'd like to clean both rooms, but give priority to the bedroom. 
\end{example}

\subsection{Defining Optimal Strategies}

Henceforth, we consider a \emph{constrained synthesis \LTLfF problem} \constrainedproblem{},
defined as before %
except that now $\varphi$ is an \LTLfF\ formula.
Such formulae assign satisfaction values to traces, 
allowing us to rank traces according to the extent to which they satisfy the expressed preferences. 
It remains to lift this preference order to strategies. 

Perhaps the most obvious way to rank strategies is to consider the minimum value of any trace induced by the strategy, 
preferring strategies that can guarantee the highest worst-case value. This is the approach adopted by \cite{AlmagorBK16} for \LTLF\ synthesis. We formalize it for constrained \LTLfF\ synthesis as follows:

{\color{purple}

}

\begin{definition}\label{def:bgv}
The \emph{best guaranteed value} of strategy $\strategy$, denoted $\wcv{\strategy}$, is the minimum value of $\eval{\trace}{\varphi}{}$ over all %
$\trace \in \fulltraces{\strategy}$ (or 0 if $\fulltraces{\strategy} = \emptyset$). 
A strategy $\strategy$ is \emph{\opt\ w.r.t.\ $(\assume,\varphi)$} if it is a $\assume$-strategy and no $\assume$-strategy $\strategy'$ exists with $\wcv{\strategy'} > \wcv{\strategy}$. 
\end{definition}

Optimizing for the best guaranteed value seems natural, %
but can be insufficiently discriminative. 
Consider a simple scenario with $\C{X}=\{x\}$ and $\C{Y}=\{y\}$.
If the environment plays $x$, then we get value $0$ no matter what, 
and if $\neg x$ is played, a value of $1$ is achieved by playing $y$, and $0$ if $\neg y$ is played. 
Clearly, we should prefer to play $y$ after $\neg x$, %
yet the strategy that plays $\neg y$ following $\neg x$ is \opt, 
since like every strategy, its bgv is 0. 
This motivates us to introduce a stronger, context-aware, notion of optimality:

\begin{definition}
Given %
a strategy $\strategy$, trace $\trace \in \compat{\strategy}$ %
that does not contain $\endtrace$, and $X \in 2^\Uvars$,
the \emph{best guaranteed value of $\strategy$ starting from $\trace  \! \cdot \! X$}, written $\wcvstart{\strategy}{\trace,X}$, 
is the minimum of $\eval{\trace'}{\varphi}{}$ over all traces 
$\trace' \in \fulltraces{\strategy}$ such that 
$\pi \cdot (X\cup Y) \sqsubseteq \pi'$
for some $Y \in \Cvars$ (or 0 if no such trace exists). 
A strategy $\strategy$ is a \emph{strongly \opt\ w.r.t.\ $(\assume,\varphi)$} if it is an $\assume$-strategy, and there is no $\assume$-strategy $\strategy'$, trace $\trace \in \compat{\strategy} \cap \compat{\strategy'}$ without $\endtrace$, 
and $X \in 2^\Uvars$
such that $\wcvstart{\strategy'}{\trace,X} > \wcvstart{\strategy}{\trace,X}$.
\end{definition}

Strongly \opt strategies take advantage of any favorable situation during execution to improve the best worst-case value. In the preceding example, they allow us to say that the first strategy is better than the second. 

\section{Algorithms: High-Quality \LTLf Synthesis}
\label{sec:alg-high-qual-synth}
In this section, we present novel techniques to compute \opt %
and strongly \opt strategies for 
a constrained \LTLfF\ synthesis problem $\constrainedproblem{}$. 
As in Section \ref{dba-construction}, we focus on the case where $\assume$ is a conjunction $\assume_s \wedge \assume_c$ of safe and co-safe formulae.

\subsection{Automaton for \LTLfF}\label{ss:aut-ltlff}
It has been shown in \cite{AlmagorBK16} how to construct, 
for a given \LTLfF\ formula $\varphi$ and set of values $\desiredvals \subseteq [0,1]$, 
an \NFA $\C{A}_{\varphi, 
\desiredvals} = \langle \alphPV, Q, \delta, Q_0, F\rangle$
that accepts finite traces $\pi$ with $\eval{\trace}{\varphi}{} \in \desiredvals$. 
We briefly recall the construction here.
We denote by $sub(\varphi)$ the set of subformulas of~$\varphi$,
and let $C_\varphi$ be the set of functions $g: sub(\varphi) \rightarrow [0,1]$ such that $g(\psi) \in V(\psi)$ for all $\psi \in sub(\varphi)$.
$Q$ contains all \emph{consistent} functions in $C_\varphi$, where 
a function $g$ is consistent if, for every $\psi \in sub(\varphi)$, the following hold:
\begin{itemize}[noitemsep,topsep=0pt]
    \item if $\psi = \top$, then $g(\psi) = 1$, and if $\psi = \bot$, then $g(\psi) = 0$
    \item if $\psi \in \C{P}$ then $g(\psi) \in \{0,1\}$
    \item if $\psi = f(\psi_1,\ldots,\psi_k))$, then $g(\psi) = f(g(\psi_1), \ldots, g(\psi_k))$
\end{itemize}
The transition function $\delta$ is such that $g' \in \delta(g,\sigma)$ whenever:
\begin{itemize}[noitemsep,topsep=0pt]
    \item $\sigma = \set{p \in \C{P} \mid g(p) = 1}$
    \item $g(\ltlnext \psi_1) = g'(\psi_1)$ for every $\ltlnext \psi_1 \in sub(\varphi))$
    \item $g(\ltluntil{\psi_1}{\psi_2}) = \max\set{g(\psi_2), \min\set{g(\psi_1), g'(\ltluntil{\psi_1}{\psi_2})}}$ for every $\ltluntil{\psi_1}{\psi_2} \in sub(\varphi)$
\end{itemize}
Finally, the set of initial states is $Q_0 = \set{q \in Q \mid g(\varphi) \in \desiredvals}$, and 
$F = \set{ g \mid g(\psi_2) = g(\ltluntil{\psi_1}{\psi_2}) \text{ for all } \ltluntil{\psi_1}{\psi_2} \in sub(\varphi)} \cap \set{ g \mid g(\ltlnext{\psi}) = 0 \text{ for all } \ltlnext{\psi} \in sub(\varphi)}$. %

The NFA $\C{A}_{\varphi,\desiredvals}$ can be constructed in single-exponential time, 
and  $\C{L}(\C{A}_{\varphi,\desiredvals}) = \{\pi \mid \eval{\trace}{\varphi}{} \in \desiredvals\}$  \cite{AlmagorBK16}.
By determinizing $\C{A}_{\varphi,\desiredvals}$, we obtain a DFA $\hat{\C{A}}_{\varphi,\desiredvals}$
that accepts the same language and can be constructed in double-exponential time. In what follows, $\desiredvals$ will always take the form $[b,1]$,
so we'll use $\hat{\C{A}}_{\varphi \geq b}$  in place of  $\hat{\C{A}}_{\varphi, [b,1]}$. 

\subsection{Synthesis of \opt strategies} \label{ssec:algo-bgv}
We describe how to construct a \opt strategy.
First note that given $b \in [0,1]$, we can construct a DBA $\C{A}_{\varphi \geq b}^\assume$ that recognizes traces
such that either (i) $\assume =\assume_s \wedge \assume_c$ is violated and $\endtrace$ does not occur, %
or (ii) $\endtrace$ occurs exactly once and the induced finite trace $\pi$ is such that $\eval{\trace}{\varphi}{} \geq b$. 
Indeed, we simply reuse the construction from Section \ref{dba-construction}, replacing the DFA $\C{A}_{g}$ 
with the DFA  $\hat{\C{A}}_{\varphi \geq b}$. %
We next observe that an $\assume$-strategy $\strategy$ with $\wcv{\strategy} \geq b$ exists iff 
the DBA game $\langle \C{X}, \C{Y} \cup \{\endtrace\}, \C{A}_{\varphi \geq b}^\assume \rangle$ is winning.
Thus, by iterating over the values in $V(\varphi)$ in descending order, we can determine the maximal $b^*$
for which an $\assume$-strategy $\strategy$ with $\wcv{\strategy} \geq b^*$ exists. 
A \opt\ strategy can be computed by constructing a winning strategy for the \DBA game $\langle \C{X}, \C{Y} \cup \{\endtrace\}, \C{A}_{\varphi \geq b^*}^\assume \rangle$, 
using the approach in Section \ref{ssec:dba-games}. As there are only exponentially many values in 
$V(\varphi)$ (Prop.\ \ref{exp-values}), the overall construction takes double-exponential time.

\begin{theorem}\label{thm:opt_strategies}
A \opt strategy can be constructed in double-exponential time. %
\end{theorem}

\def\bgval{v}
\def\optval{v^+}
\newcommand\autval[1]{\C{A}_{\geq #1}^\assume}
\newcommand\winoptgen{
Win(\bgvaut,\optaut)}
\def\bgvaut{\C{A}_{\varphi \geq \bgval}^\assume}
\def\optaut{\C{A}_{\varphi \geq \optval}^\assume}

\newcommand\winoptgeni[1]{Win^{#1}(\bgvaut,\optaut)}

\subsection{Synthesis of strongly \opt strategies}
To compute a strongly \opt stategy, we build a transducer that runs in parallel DBAs $\C{A}_{\geq b}^\assume$ for different values~$b$,
and selects outputs symbols so as to advance within the `best' applicable winning region. 
This idea can be formalized as follows. As in Section \ref{ssec:algo-bgv}, 
we first determine the maximal $b^* \in V(\varphi)$ 
for which a $\assume$-strategy $\strategy$ with $\wcv{\strategy} \geq b^*$ exists, and set 
$B= V(\varphi) \cap [b^*,1]$. %
In the process, we will compute, for each $b \in B$, the sets $Win(\C{G}_{b})$ and $RA(\C{G}_{b})$
for the DBA game $\C{G}_{b} = \langle \C{X}, \C{Y}, \C{A}_{\geq b}^\assume \rangle$ based on the DBA $\autval{b}=\langle 2^{\C{P}\cup \{\endtrace\}}, Q_b, \delta_b, q_0^{b}, F_{b}\rangle$. 
In the sequel, we will assume that the elements of $B$ are ordered as follows: $b_1 < b_2 \ldots < b_m$, with $b^* = b_1$ and $b_m=1$.

We now proceed to the definition of the desired transducer $\C{T}^\mathsf{str}=\langle 2^\C{X}, 2^{\C{Y} \cup \{\endtrace\}}, Q_\mathsf{str}, \delta_\mathsf{str}, \omega_\mathsf{str}, q_0^\mathsf{str} \rangle$, %
obtained by taking the cross product of
the set of DBAs $\autval{b}$ with $b \in B$,
in order to keep track of the current states in these automata:
\begin{itemize}[noitemsep,topsep=0pt]
\item $q_0^\mathsf{str} = (q_0^{b_1}, q_0^{b_2}, \ldots, q_0^{b_m})$  and $Q_\mathsf{str} = Q_{b_1} \times \ldots \times Q_ {b_m}$
\item $\delta_\mathsf{str}((q_1, \ldots, q_m), X)= (\delta_{b_1}(q_1,X \cup Y), \ldots, \delta_{b_m}(q_m,X \cup Y))$, where $Y= \omega_\mathsf{str}((q_1, q_2, \ldots, q_m),X)$
\end{itemize}
After reading $X$, the output function identifies the maximal value $b \in B$ such that 
current state $q_b$ of $\autval{b}$ can transition, via some symbol $X \cup Y$, into a state in $Win(\C{G}_{b})$,
and it returns the same output as the transducer $\C{T}_{\C{G}_b}$ in state $q_b$:
\begin{itemize}[noitemsep,topsep=0pt]
\item $\omega_\mathsf{str}((q_1, \ldots, q_m), X)= \omega_b(q_b,X)$, where $b= \mathsf{max}(\{v \in B\mid \exists Y\, \delta_v(q_v,X \cup Y) \in Win(\C{G}_{v})\})$
\end{itemize}
We note that the transducer $\C{T}_{\C{G}_b}$ can be defined as in Section \ref{ssec:dba-games} even when $q_0^b \not \in Win(\C{G}_{b})$,
but it only returns `sensible' outputs when it transitions to $Win(\C{G}_{b})$. 

\begin{theorem}\label{thm:strong-opt}
$\C{T}^\mathsf{str}$ implements a strongly \opt strategy and can be constructed in double-exponential time. %
\end{theorem}

\section{Discussion and Concluding Remarks}
\label{sec:conclusion}
It has been widely remarked in the (infinite) LTL synthesis literature that environment assumptions are ubiquitous: the existence of winning strategies is almost always predicated on some kind of environment assumption. This was observed in the work of \cite{ChatterjeeHJ08}, %
motivating the introduction of 
the influential assume-guarantee synthesis model, and in work on rational synthesis \cite{fis-kup-lus-tacas10}, where the environment is assumed to act as a rational agent, and synthesis necessitates 
finding a Nash equilibrium. Interesting reflections on the role of assumptions in \LTL synthesis, together with a survey of the relevant literature, can be found in \cite{DBLP:journals/corr/BloemEJK14}.

In this paper, we explored the issue of handling environment assumptions in \LTLf synthesis \cite{DBLP:conf/ijcai/GiacomoV13}, the counterpart of \LTL synthesis for programs that terminate. Our starting point %
was the observation that the standard approach to handling assumptions in \LTL synthesis (via logical implication) fails in the finite-trace setting. %
This led us to propose an extension of \LTLf\ synthesis %
 that explicitly accounts for environment assumptions.
The key insight underlying the new model of \emph{constrained \LTLf\ synthesis} 
is that while the synthesized program must realize the objective in a finite number of steps, the environment continues to exist after the program terminates, so environment assumptions should be interpreted under infinite \LTL\ semantics. 

We studied the relationships holding between constrained \LTLf\ synthesis and (standard) \LTLf\ and \LTL\ synthesis.
In particular, we identified a fundamental difficulty in reducing constrained \LTLf\ synthesis to \LTLf\ synthesis --  the former problem can require unbounded strategies, while bounded strategies suffice for the latter.  Nevertheless, when the constraints were restricted to the safe \LTL fragment, a reduction from constrained \LTLf synthesis to \LTLf synthesis is possible. 
Interestingly, this explains why planning -- more naturally conceived as a constrained \LTLf synthesis problem -- can also be encoded as \LTLf synthesis.
The connection between synthesis and planning  has been remarked in several works (see e.g., \cite{deg-var-ijcai15,DIppolitoRS18,cam-bai-mui-mci-icaps18,cam-bai-mui-mci-ccai18,cam-mui-bai-mci-ijcai18}).  
We also showed how to reduce constrained \LTLf synthesis to (infinite) \LTL synthesis, which provides a worst-case optimal means of solving constrained \LTLf synthesis problems, in the general case, using (infinite) \LTL synthesis tools. %
In the case where our constraint is comprised of a conjunction of safe and co-safe formulae, 
we showed that the constrained \LTLf synthesis problem can be reduced to DBA games
and the winning strategy determined from the winning region. %
What makes our approach interesting is that 
the DBA is constructed via manipulation of DFAs, much easier to handle in practice than infinite-word automata.

We next turned our attention to the problem of 
augmenting constrained \LTLf synthesis with quality measures. We were motivated by practical concerns surrounding the ability to differentiate and synthesize high-quality strategies in settings where we may have a collection of mutually unrealizable objective formulae and  %
alternative strategies of differing quality.  Our work builds on results for the infinite case, e.g., \cite{AlmagorBK16,AlmagorKRV17,kupferman-qualsurvey16} with and without environment assumptions. 
We adopted \LTLfF %
as our language for specifying quality measures. While the syntax of \LTLfF  is utilitarian, many more compelling preference languages are reducible to this core language.  We defined two different notions of optimal strategies -- bgv-optimal and strongly bgv-optimal. The former adapts a similar definition in \cite{AlmagorBK16} and the latter originates with us.  
We focused again on %
assumptions that can be expressed as conjunctions of safe and co-safe formulae and 
provided algorithms to compute bgv- and strongly bgv-optimal strategies with optimal (2EXP) complexity. %

Proper handling of environment assumptions and quality measures, together with the design of efficient algorithms for such richer specifications, is essential to putting \LTLf synthesis into practice. 
The present paper makes several important advances in this direction and 
also suggests a number of interesting topics for future work including:
the study of other types of assumptions in the finite-trace setting (e.g. rational synthesis), 
the exploitation of more compelling KR languages for specifying preferences, 
and the exploration of further ways of comparing and ranking strategies  
(perhaps incorporating notions of cost or trace length).

\subsection*{Relation to Conference Version}
This paper appears in the Proceedings of the 16th International Conference on Knowledge Representation and Reasoning (KR 2018) without the appendix proofs. The body of this paper is the same as the KR 2018 paper, except that a minor typographic error in the translation of \LTL into \LTLf (the definition of $\translate(\ltluntil{\varphi_1}{\varphi_2}) $ on page 6), which originally appeared in \cite{DBLP:conf/ijcai/GiacomoV13} and was repeated in the KR 2018 paper, has been corrected.

        \subsection*{Acknowledgements}
    This work was partially funded by the Natural Sciences and Engineering Research Council of Canada (NSERC) and the ANR project GoAsQ (ANR-15-CE23-0022).

    \begin{small}
  \bibliography{main}
  \bibliographystyle{aaai}
    \end{small}
\section*{Proofs}

\medskip

\noindent\textbf{Theorem \ref{thm:safe_reduction}}.
When $\assume$ is a safe formula, 
the constrained \LTLf\ specification $\C{S}=\constrainedproblem{}$ is realizable iff the \LTLf specification $\C{S}'=\langle \Uvars, \Cvars, \assume' \rightarrow \varphi \rangle$ is realizable,
where $\assume'$ is %
obtained from $\assume$ by replacing every occurrence of $\ltlnext{\psi}$ by $\ltlweaknext{\psi}$. 

\begin{proof}
Let $\strategy'$  be a winning strategy for $\C{S}'=\langle \Uvars, \Cvars, \assume' \rightarrow \varphi \rangle$,  
with $\assume$ a safe formula. We define a strategy $\sigma$ for $\C{S}=\constrainedproblem{}$ by setting $\sigma(X_1\cdots X_n)$ equal to
\begin{itemize}
\item $\sigma'(X_1 \cdots X_n) \setminus \{\endtrace\}$, when $\endtrace \in \sigma'(X_1 \cdots X_n)$ and \newline $(X_1 \cup \strategy'(X_1)) \ldots (X_n \cup \strategy'(X_1\cdots X_n)) \not \models \assume'$; 
\item $\sigma'(X_1 \cdots X_n)$, otherwise.
\end{itemize}
To show that $\sigma$ is a winning strategy, take some $\mathbf{X} \in (2^\C{X})^\omega$. Then $n_{\strategy',\mathbf{X}} < \infty$, %
and %
 the finite trace $\pi' = \tracestrat{\sigma'}{\mathbf{X}}$ is such that %
(i) $\pi'\models \neg \assume'$, or (ii) $\pi' \models \varphi$
If (i) holds, then $n_{\strategy,\mathbf{X}}=\infty$ (since we will remove $\endtrace$),
and $\pi = \inftracestrat{\sigma}{\mathbf{X}}$ contains $\pi'$ as a prefix. 
We can then use the \LTLf\ equivalence $\neg \ltlweaknext{} \varphi \equiv \ltlnext{} \neg \varphi$ and the fact that $\assume$ is safe to derive
$\pi \models \neg \assume$. If (ii) holds (and (i) does not), then it follows from the definition of $\sigma$ that $n_{\strategy,\mathbf{X}}=n_{\strategy',\mathbf{X}}$ and $\pi = \tracestrat{\sigma}{\mathbf{X}} =\tracestrat{\sigma'}{\mathbf{X}}$, so $\pi \models \varphi$. 

\smallskip

For the other direction, let $\strategy$ be a winning strategy for $\C{S}$. 
Define a strategy $\sigma'$ for $\C{S}'$ %
by setting $\sigma'(X_1\cdots X_n)$ equal to
\begin{itemize}
\item $\strategy(X_1\cdots X_n) \cup \set{\endtrace}$, if $(X_1 \cup \strategy(X_1)) \ldots (X_n \cup \strategy(X_1\cdots X_n))$ is a bad prefix for $\assume$,
and 
$\endtrace \not \in \sigma'(X_1\cdots X_k)$
for 
$k < n$;
\item $ \strategy(X_1\cdots X_n) \setminus \set{\endtrace}$, if 
$\endtrace \in \sigma'(X_1\cdots X_k)$ for some $k < n$;
\item $\sigma'(X_1\cdots X_n)= \strategy(X_1\cdots X_n)$, otherwise
\end{itemize}

Basically, $\sigma'$ mimics $\sigma$ and outputs $\endtrace$ as soon as a bad prefix is reached or $\sigma$ terminates the execution due to satisfaction of $\varphi$, whichever situation occurs first.
The construction of $\sigma'$ is such that executions always terminate
-- which is not necessarily true for $\sigma$ -- and induce finite traces that satisfy the \LTLf formula $\alpha' \rightarrow \varphi$.

To show that $\sigma'$ is a winning strategy for $\C{S}'$, take some $\mathbf{X} \in (2^\C{X})^\omega$.  Consider the infinite induced trace $\pi = \inftracestrat{\sigma}{\mathbf{X}}$, and let $k_{\mathbf{X}}$
be the minimum $k$ such that $\pi(1) \ldots \pi(k)$ is a bad prefix for $\assume$ ($k_{\mathbf{X}} = \infty$ if $\pi$ has no bad prefix). 
We examine these three cases separately: 
(i) $n_{\strategy,\mathbf{X}} < \infty$ and $k_{\mathbf{X}} \leq n_{\strategy,\mathbf{X}}$, 
(ii) $n_{\strategy,\mathbf{X}} < \infty$ and $k_{\mathbf{X}} > n_{\strategy,\mathbf{X}}$, and
(iii) $n_{\strategy,\mathbf{X}}=\infty$.
 
If (i) holds,
then $n_{\strategy',\mathbf{X}} = k_{\mathbf{X}}$, and  $\pi' = \tracestrat{\sigma'}{\mathbf{X}}= \pi(1) \ldots \pi(k_\mathbf{X})$.
From the fact that $\pi(1) \ldots \pi(k_{\mathbf{X}})$ is a bad prefix for the safe formula $\assume$ and the \LTLf\ equivalence $\neg \ltlweaknext{} \varphi \equiv \ltlnext{} \neg \varphi$, 
we obtain %
$ \pi'\models \neg \assume'$.  

If (ii) holds, then
it follows from the definition of $\sigma'$ that
$n_{\strategy',\mathbf{X}}=n_{\strategy,\mathbf{X}}$ and 
$\tracestrat{\sigma'}{\mathbf{X}}=\pi(1) \ldots \pi(n_{\strategy,\mathbf{X}})$. Since 
$\sigma$ is a winning strategy for $\C{S}$ with $n_{\strategy,\mathbf{X}} < \infty$, we must have $\pi(1) \ldots \pi(n_{\strategy,\mathbf{X}}) \models \varphi$, hence $\tracestrat{\sigma'}{\mathbf{X}} \models \varphi$.

If (iii) holds,
then it must be the case that  $\pi \not\models \alpha$ (because $\varphi$ is a winning strategy).
As $\alpha$ is a safe formula, $\pi$ must contain a bad prefix for $\alpha$, so $k_{\mathbf{X}} < \infty$. Using a similar argument as in case (i), we can show that $\pi' \models \neg \alpha'$.
\end{proof}

\medskip

\noindent\textbf{Theorem \ref{reduction-infinite}}
The constrained \LTLf\ specification $\C{S}=\constrainedproblem{}$ is realizable iff \LTL\ specification $\C{S}^\infty=\langle \Uvars, \Cvars \cup \set{\mathsf{alive},\endtrace}, \psi_{\assume, \varphi} \rangle$ is realizable. Moreover, 
for every winning strategy $\sigma$ for $\C{S}^\infty$, 
the strategy $\sigma'$ defined by $\strategy'(X_1\cdots X_n) \defeq \strategy(X_1 \cdots X_n) \setminus \set{\mathsf{alive}}$ is a winning strategy for $\C{S}$. 
\begin{proof} 
For the first direction, suppose that $\C{S}=\constrainedproblem{}$ is realizable, and let $\sigma$ be a winning strategy for $\C{S}$. 
By definition, for every $\mathbf{X} \in (2^\C{X})^\omega$, the trace $\inftracestrat{\sigma}{\mathbf{X}}$ satisfies one of the following: $(i)$ it has a finite prefix of length $n_{\strategy,\mathbf{X}} < \infty$ that satisfies $\varphi$, or $(ii)$  $n_{\strategy,\mathbf{X}} = \infty$ and  $\inftracestrat{\sigma}{\mathbf{X}} \not \models \assume$. 
We define a strategy $\sigma^\infty$ for $\C{S}^\infty=\langle \Uvars, \Cvars \cup \set{\mathsf{alive},\endtrace}, \psi_{\assume, \varphi} \rangle$
as follows: %
\begin{itemize}
\item $\sigma^\infty(X_1\cdots X_n)= \sigma(X_1\cdots X_n) \cup \{\mathsf{alive}\}$, if there is no $k < n$ such that $\endtrace \in \sigma(X_1\cdots X_k)$  
\item $\sigma^\infty(X_1\cdots X_n) = \sigma(X_1\cdots X_n)$, otherwise
\end{itemize}
We claim that $\sigma^\infty$ is a winning strategy for $\C{S}^\infty$. Take some $\mathbf{X}=X_1 X_2 \ldots $, let $\pi=\inftracestrat{\sigma}{\mathbf{X}}$ and $\pi^\infty= \inftracestrat{\sigma^\infty}{\mathbf{X}}$. 
We first show that $\pi^\infty \models \psi_{\endtrace}$. First note that if $\endtrace \in \pi^\infty(i)$,  then $\endtrace \in \pi(i)$, which means $\endtrace \not \in \pi(i+1)$ (since $\pi$ contains at most one $\endtrace$).
It follows that $\mathsf{alive} \in \pi^\infty(i)$ and $\mathsf{alive}\not \in \pi^\infty(i+1)$. Next suppose that $\mathsf{alive} \in \pi^\infty(i)$ but $\mathsf{alive} \not \in \pi^\infty(i+1)$. 
This can only occur if $\endtrace \in \pi(i)$, which implies that $\endtrace \in \pi^\infty(i)$.
We have thus shown that $\pi^\infty \models \ltlalways{(\endtrace \leftrightarrow \mathsf{alive} \land \ltlnext{\neg \mathsf{alive}})}$.
We also have $\pi^\infty \models  \ltlalways{(\endtrace \rightarrow \ltlnext{\ltlalways{ \neg \endtrace}})}$, since $\endtrace \in \pi^\infty(i)$ iff $\endtrace \in \pi(i)$,
and $\pi$ contains at most one occurrence of $\endtrace$. 

We next show that $\pi^\infty  \models \left ((\assume \vee \ltleventually{\endtrace}) \rightarrow \varphi_{\mathsf{inf}}\right)$. 
First consider the case (i) where $\pi$ has a finite prefix of length $n_{\strategy,\mathbf{X}} < \infty$ that satisfies $\varphi$. 
Then $\endtrace \in \pi^\infty(n_{\strategy,\mathbf{X}})$, so 
$\mathsf{alive} \in \pi^\infty(i)$ for $1 \leq i \leq n_{\strategy,\mathbf{X}}$ 
and $\mathsf{alive} \not \in \pi^\infty(i)$ for $i > n_{\strategy,\mathbf{X}}$. 
It follows that $\pi^\infty \models \varphi_{\mathsf{inf}}$. Next suppose that (ii) holds, i.e.\ 
$n_{\strategy,\mathbf{X}} = \infty$ and  $\pi \not \models \assume$. 
Since $\assume$ only involves variables from $\C{X} \cup \C{Y}$, and $\pi^\infty$ coincides with $\pi$ on 
$\C{X} \cup \C{Y}$, it follows that $\pi^\infty  \not \models \assume$. 
As $n_{\strategy,\mathbf{X}} = \infty$, we know that $\endtrace$ does not occur in $\pi$. The same must hold for $\pi^\infty$,
hence $\pi^\infty \not \models \ltleventually{\endtrace}$.    
We thus obtain $\pi^\infty  \models \neg (\assume \vee \ltleventually{\endtrace})$, hence 
$\pi^\infty  \models \left ((\assume \vee \ltleventually{\endtrace}) \rightarrow \varphi_{\mathsf{inf}}\right)$. 

\smallskip

For the other direction, suppose $\C{S}^\infty$ is realizable, and let $\sigma^\infty$ be a winning strategy for $\C{S}$. 
We define $\sigma$ as follows:  $\sigma(X_1 \ldots X_n)= \sigma^\infty(X_1 \ldots X_n) \setminus \{\mathsf{alive}\}$. 
Our aim is to show that $\sigma$ is a winning strategy for $\C{S}$. Consider some $\mathbf{X}=X_1 X_2 \ldots $, and 
let $\pi^\infty= \inftracestrat{\sigma^\infty}{\mathbf{X}}$ and $\pi=\inftracestrat{\sigma}{\mathbf{X}}$. 
First note that since $\pi^\infty \models \ltlalways{(\endtrace \rightarrow \ltlnext{\ltlalways{ \neg \endtrace}})}$, 
$\pi^\infty$ contains at most one occurrence for $\endtrace$, and the same holds for $\sigma$. 
It follows that $\sigma$ is a valid strategy. Next suppose that $\endtrace \in \pi^\infty(n)$.
Then $\tracestrat{\sigma}{\mathbf{X}}= \pi(1) \ldots \pi(n)$. 
As $\pi^\infty  \models \left ((\assume \vee \ltleventually{\endtrace}) \rightarrow \varphi_{\mathsf{inf}}\right)$
and $\pi^\infty \models \ltleventually{\endtrace}$, we have $\pi^\infty \models \varphi_{\mathsf{inf}}$.
When combined with $\pi^\infty \models \ltlalways{(\endtrace \leftrightarrow \mathsf{alive} \land \ltlnext{\neg \mathsf{alive}})}$,
we get that $\mathsf{alive} \in \pi^\infty(i) $ for $i \leq n$ and $\mathsf{alive} \not \in \pi^\infty(i) $ for $i > n$. 
This, together with the definition of $\varphi_{\mathsf{inf}}$, implies that the induced finite trace $\pi(1) \ldots \pi(n)$
satisfies $\varphi$. Finally, consider the case where $\endtrace$ does not occur in $\pi^\infty$. 
We know that $\pi^\infty  \models \left ((\assume \vee \ltleventually{\endtrace}) \rightarrow \varphi_{\mathsf{inf}}\right)$.
Note that we cannot have $\pi^\infty  \models  \varphi_{\mathsf{inf}}$ since it implies, when combined with 
$\pi^\infty  \models  \ltlalways{(\endtrace \leftrightarrow \mathsf{alive} \land \ltlnext{\neg \mathsf{alive}})}$,
that $\endtrace$ occurs in $\pi^\infty$. It follows that $\pi^\infty  \models \neg (\assume \vee \ltleventually{\endtrace}) $,
or equivalently, $\pi^\infty  \models \neg \assume \wedge \ltlalways{\neg \endtrace}$. From this, we can derive that $\pi \not \models \assume$
and $\pi$ does not contain $\endtrace$. We have thus established that $\sigma$ is a winning strategy.
\end{proof}

\medskip

\noindent\textbf{Theorem \ref{dba-thm}}
The \DBA $\constraintdba$ accepts infinite traces $\pi$ such that either: (i) 
$\pi(1)\cdots\pi(n) \models \varphi$ and $\endtrace$ occurs only in $\pi(n)$,
or (ii) $\pi \not\models \assume_s \wedge \assume_c$ and $\endtrace$ does not occur in $\pi$. 
$\constraintdba$ can be constructed in double-exponential time in 
$|\assume_s| + |\assume_c| + |\varphi|$.
\begin{proof}
First, we prove that the language of $\constraintdba$ contains the set of traces that satisfy one of the conditions $(i)$ and $(ii)$. Then, we prove that the language of $\constraintdba$ is contained in the set of traces that satisfy one of the conditions $(i)$ and $(ii)$.
The construction of $\constraintdba$ is polynomial in the size of $\C{A}_s$, $\C{A}_c$, and $\C{A}_g$, which can be constructed in double-exponential time in $|\assume_s|$, $|\assume_c|$, and $|\assume_g|$, respectively.

$(\supseteq)$
(i) 
If $\pi(1)\cdots\pi(n) \models \varphi$ and $\endtrace$ occurs only in $\pi(n)$,
then the run of $\C{A}_g$ on $\pi(1)\cdots\pi(n)$ is accepting, i.e. its last state belongs to $F_g$. 
It then follows from the definition of $\constraintdba$ that the run of $\constraintdba$ on $\pi$ will transition to $q_\top$ after reading $\pi(n)$ and, because $\endtrace$ occurs only in $\pi(n)$, it will then loop at $q_\top$. %
As $q_\top$ is an accepting state of $\constraintdba$, this shows that $\constraintdba$ accepts $\pi$.
(ii)
If $\pi \not\models \assume_s \wedge \assume_c$ and $\endtrace$ does not occur in $\pi$,
then either  $\pi \not\models \assume_s$ or  $\pi \not\models \assume_c$. 
In the first case ($\pi \not\models \assume_s$), the trace $\pi$ contains a bad prefix $\pi(1)\cdots\pi(n)$, and w.l.o.g. we can suppose that this is the shortest such prefix. 
It follows that the run of $\C{A}_s$ on $\pi(1)\cdots\pi(n)$ is accepting. From the definition of $\constraintdba$, the run of $\constraintdba$ on $\pi$ transitions to a state $(q_s,q_c,q_g)$ where $q_s = q_\mathsf{bad}$. Because $\endtrace$ does not occur in $\pi$, after having read $\pi(n)$, the run of $\constraintdba$ on $\pi$ will remain among the states whose third component is $q_\mathsf{bad}$. By construction, these states are accepting, so $\constraintdba$ accepts $\pi$.
In the second case ($\pi \not\models \assume_g$), the run of $\C{A}_c$ on finite prefix $\pi(1)\cdots\pi(n)$ is not accepting for any $n < \infty$. 
It follows that the run of $\constraintdba$ on $\pi$ will not visit any state $(q_s,q_c,q_g)$ with $q_c = q_\mathsf{good}$, and since $\pi$ does not contain $\endtrace$, it also cannot contain the state $q_\bot$. Thus, the run of $\pi$ will only visit accepting states, so $\constraintdba$ accepts $\pi$.

$(\subseteq)$
Let $\pi$ be an infinite trace that is accepted by $\constraintdba$. We distinguish three cases: $(a)$ $\endtrace$ does not occur in $\pi$; $(b)$ $\endtrace$ occurs exactly one time in $\pi$; $(c)$ $\endtrace$ occurs more than one time in $\pi$.
Case $(a)$:
if $\endtrace$ does not occur in $\pi$, then $q_\top$ does not occur in the run of $\constraintdba$ on $\pi$. As this run is accepting but does not contain $q_\top$, %
it must either hit infinitely often states with $q_\mathsf{bad}$, or hit infinitely often states without $ q_\mathsf{good}$. 
In the first case, let $\pi(1)\cdots\pi(n)$ be the smallest prefix of $\pi$ such that after reading $\pi(1)\cdots\pi(n)$, the DBA $\constraintdba$ enters a state $(q_s,q_c,q_g)$ with $q_s = q_\mathsf{bad}$.
By construction of $\constraintdba$, it must be that the run of $\C{A}_s$ on $\pi(1)\cdots\pi(n)$ is accepting. 
Thus, $\pi(1)\cdots\pi(n)$ %
 is a bad prefix of $\assume_s$, which means  $\pi \not\models \assume_s$.
In the second case, we observe that if a run enters a state $(q_s,q_c,q_g)$ with $q_c = q_\mathsf{good}$, then it remains in a state with $\mathsf{good}$ in the second component unless a symbol with $\endtrace$ is read. As the considered trace $\pi$ does not contain $\endtrace$, it follows that 
the run of $\constraintdba$ on $\pi$ does not hit any state $(q_s,q_c,q_g)$ with $q_c = q_\mathsf{good}$. Hence, 
$\C{A}_c$ must not accept any finite prefix of $\pi$, so $\pi$ does not contain any good prefixes for $\assume_c$, i.e.\ 
$\pi \not\models \assume_c$.
This concludes that case $(a)$ implies case $(ii)$.
Case $(b)$: suppose that $\endtrace$ occurs exactly one time in $\pi$, say in the $n$th symbol $\pi(n)$.
Immediately after reading $\pi(n)$, the DBA $\constraintdba$ will transition to either $q_\top$ or $q_\bot$. However, since the  run of $\constraintdba$ on $\pi$ is accepting, it cannot contain $q_\bot$. We thus have a transition of the form $\delta((q_s,q_c,q_g), \pi(n))=q_\top$ with
$\delta_g(q_g,\pi(n)) \in \QF{g}$. 
It follows that the finite prefix $\pi(1)\cdots\pi(n)$ is accepted by the DFA $\C{A}_g$, i.e., $\pi(1)\cdots\pi(n) \models \varphi$. 
We have thus shown that case $(b)$ implies case $(i)$.
To conclude the proof, observe that Case $(c)$ is not possible. 
Indeed, the first symbol with $\endtrace$ forces an automaton transition to either $q_\top$ or $q_\bot$, 
and the second symbol with $\endtrace$ forces an automaton transition to $q_\bot$.
Because state $q_\bot$ is absorbing and not accepting, 
$\constraintdba$ does not accept traces where $\endtrace$ occurs more than one time.
\end{proof}

\medskip

\noindent\textbf{Theorem \ref{thm:correct-constrained}}
Consider a constrained \LTLf specification $\C{S}=$\agproblem{\assume_s \wedge \assume_c}{\varphi} where $\assume_s$ (resp.\ $\assume_c$) is a safe (resp.\ co-safe) formula. 
Then:
\begin{itemize}
\item  $\C{S}$ is realizable iff the \DBA game $\C{G}^{\assume_s, \assume_c}_{\varphi}$ is winning;
\item Every winning strategy for 
$\C{G}^{\assume_s, \assume_c}_{\varphi}$ is a winning strategy for $\C{S}$, and vice-versa;
\item Deciding whether $\C{G}^{\assume_s, \assume_c}_{\varphi}$ is winning, and constructing a winning strategy when one exists,
can be done in 2EXP. 
\end{itemize}
\begin{proof}
By Theorem \ref{dba-thm},
the language of the \DBA $\constraintdba$ contains all, and only, the infinite traces 
$\pi$ such that either:
(i) 
$\pi(1)\cdots\pi(n) \models \varphi$ and $\endtrace$ occurs only in $\pi(n)$,
or (ii) $\pi \not\models \assume_s \wedge \assume_c$ and $\endtrace$ does not occur in $\pi$. 
By definition, a strategy $\sigma$ is a winning strategy for $\C{G}^{\assume_s, \assume_c}_{\varphi}$ iff 
every induced infinite trace $\inftracestrat{\sigma}{\mathbf{X}}$ is accepted by $\constraintdba$. %
It follows that winning strategies for 
$\C{G}^{\assume_s, \assume_c}_{\varphi}$ are winning strategies for $\C{S}$, and vice versa.
By Theorem \ref{dba-thm}, we can construct the DBA $\constraintdba$ in double-exponential time.
We can then construct, in polynomial time in $|\constraintdba|$ the set $Win(\C{G}^{\assume_s, \assume_c}_{\varphi} )$
and check whether  $((q_0)_{s}, (q_0)_{c}, (q_0)_{g}) \in Win(\C{G}^{\assume_s, \assume_c}_{\varphi} )$.
If the latter holds, then we construct a winning strategy in the form of a transducer. This step is also
polynomial w.r.t.\ $|\constraintdba|$, and thus the entire procedure can be performed in double-exponential time. 
\end{proof}

\medskip

\noindent\textbf{Theorem \ref{thm:opt_strategies}}
A \opt strategy can be constructed in double-exponential time. %
\begin{proof}
We first argue that the described construction yields an \opt\ strategy. We know from \cite{AlmagorBK16}
that the $\hat{\C{A}}_{\varphi,\desiredvals}$ accepts finite traces $\pi$ such that
$\eval{\trace}{\varphi}{} \geq b$. By reusing the arguments from the proof of Theorem \ref{dba-thm},
we can show that $\C{A}_{\varphi \geq b}^\assume$ recognizes infinite traces
such that either (i) $\assume =\assume_s \wedge \assume_c$ is violated and $\endtrace$ does not occur, %
or (ii) $\endtrace$ occurs exactly once and the induced finite trace $\pi$ is such that $\eval{\trace}{\varphi}{} \geq b$. It follows that an $\assume$-strategy $\strategy$ with $\wcv{\strategy} \geq b$ exists iff 
the DBA game $\langle \C{X}, \C{Y} \cup \{\endtrace\}, \C{A}_{\varphi \geq b}^\assume \rangle$ is winning.
We can use the techniques described in Section \ref{ssec:dba-games} to decide whether such a DBA is winning.
By considering the values in $V(\varphi)$ in descending order, 
we can determine the maximal $b^*$
for which an $\assume$-strategy $\strategy$ with $\wcv{\strategy} \geq b^*$ exists. 
We can then compute such a strategy, for the identified best value  $b^*$, 
by constructing a winning strategy for the \DBA game $\langle \C{X}, \C{Y} \cup \{\endtrace\}, \C{A}_{\varphi \geq b^*}^\assume \rangle$, again using the techniques from Section \ref{ssec:dba-games}.

For each value $b$, we can construct the DFA $\hat{\C{A}}_{\varphi,\desiredvals}$ in double-exponential time,
and the DFAs for safe and co-safe constraints can also be constructed in double-exponential time. 
It follows that $\C{A}_{\varphi \geq b}^\assume$ can also be constructed in double-exponential time.
Determining whether the DBA game $\langle \C{X}, \C{Y} \cup \{\endtrace\}, \C{A}_{\varphi \geq b}^\assume \rangle$ is winning, and constructing a winning strategy for the game, is also in double-exponential time. Finally, we note that
all of the preceding double-exponentional time operations are performed at most once per value $b \in V(\varphi)$, 
so the overall procedure runs in double-exponential time. 
\end{proof}

\medskip

\medskip

The next two lemmas will be used to prove Theorem \ref{thm:strong-opt}. In what follows, it will be convenient to slightly abuse notation and use $\delta(q,\pi)$, with $\pi$ a finite trace, to indicate the automata state resulting 
from reading $\pi$ starting from state $q$ (and similarly for the output function $\omega$ of transducers on a finite string $X_1 \ldots X_n$). 

\begin{lemma}\label{lem:strong1}
For every $\pi=(X_1 \cup Y_1) \ldots (X_n \cup Y_n) \in 2^{\C{X} \cup \C{Y}}$, 
$X_{n+1} \in 2^\C{X}$, and $b \in [0,1]$, the following are equivalent:
\begin{enumerate}
\item there exists an $\assume$-strategy $\strategy$ %
such that $\trace \in \compat{\strategy}$ %
and $\wcvstart{\strategy}{\trace,X_{n+1}} \geq b$
\item %
$\delta_b(q_0^b, \pi \cdot (X_{n+1}\cup Y_{n+1}))\! \in Win(\C{G}_{b})$
for some $Y_{n+1}\! \in 2^{\C{Y} \cup \{\endtrace\}}$. %
\end{enumerate}
\end{lemma}
\begin{proof} %
$(\Rightarrow)$ Suppose that $\strategy$ is an $\assume$-strategy such that $\trace \in \compat{\strategy}$ and $\wcvstart{\strategy}{\trace,X_{n+1}} \geq b$.
Set $Y_{n+1} = \sigma(X_1 \ldots X_n X_{n+1})$, and let $q^* = \delta_b(q_0^b, \pi \, (X_{n+1}\cup Y_{n+1}))$. 
Consider the DBA $\autval{b, q^*} = \langle 2^{\C{X} \cup \C{Y}\cup \{\endtrace\}}, Q_b, \delta_b, q^*, F_b \rangle $, which is the same 
as $\autval{b}$ but with $q^*$ for the initial state. 
Define a strategy $\sigma^*$ for the DBA game $\C{G}_{b}^* = \langle \C{X}, \C{Y}\cup \{\endtrace\}, \C{A}_{\geq b,q^*}^\assume \rangle$ as follows:
$$\sigma^*(X_1' \ldots X_h') = \sigma(X_1 \ldots X_{n+1} X_1' \ldots X_h')$$
We claim that $\sigma^*$ is a winning strategy for $\C{G}_{b}^*$. To see why, take any 
$\mathbf{X} \in (2^{\C{X}})^*$, and let $\pi_{\mathbf{X}}^*= \inftracestrat{\sigma^*}{\mathbf{X}}$. 
We need to show that $\pi_\mathbf{X}^*$ is accepted by $\C{A}_{\geq b, q^*}^\assume$. 
Let us consider $\pi_\mathbf{X} = \inftracestrat{\sigma}{X_1 \ldots X_{n+1} \mathbf{X}}$,
and let $s_0 s_1 s_2 \ldots $ be the infinite run of $\autval{b}$ on $\pi_\mathbf{X}$.
Note that since $\trace \in \compat{\strategy}$, we know that $\pi \sqsubseteq \pi_\mathbf{X}$,
hence either (i) $\pi_\mathbf{X}$ does not contain $\endtrace$ and $\pi_\mathbf{X} \not \models \assume$,
or (ii) $\pi_\mathbf{X}$ contains exactly one occurrence of $\endtrace$ (at some position $k \leq n+1$, since $\trace$ does not contain $\endtrace$) 
and the induced finite trace $\pi_\mathbf{X}^\mathsf{f} = \tracestrat{\sigma}{X_1 \ldots X_{n+1} \mathbf{X}}$
is such that $\eval{\pi_\mathbf{X}^\mathsf{f}}{\varphi}\ \geq b$.
It follows that $\pi_\mathbf{X}$ is accepted by $\autval{b}$, and thus
the infinite run $s_0 s_1 s_2 \ldots $ of $\autval{b}$ on $\pi_\mathbf{X}$ contains infinitely many $s_j \in F_b$.
We then observe that $s_{n+1} = q^*$. Since $\C{A}_{\geq b, q^*}^\assume$  has the same transitions as 
$\autval{b}$, it follows that $s_{n+1} s_{n+2} s_{n+3} \ldots$ is the run of $\C{A}_{\geq b, q^*}^\assume$ on $\pi_\mathbf{X}^*$. 
Note that $s_{n+1} s_{n+2} s_{n+3} \ldots$ must also contain infinitely many $s_j \in F_b$, so it is an accepting run.
We have thus shown  that $\sigma^*$ is a winning strategy for $\C{G}_{b}^*$. 

As the game $\C{G}_{b}^*$ is winning, we must have $q^* \in Win(\C{G}_{b}^*)$. 
We then remark that since $\C{G}_{b}^*$ and $\C{G}_{b}$ only differ in their initial states, the two games have precisely the same
winning regions. We thus obtain $q^* = \delta_b(q_0^b, \pi \, (X_{n+1}\cup Y_{n+1})) \in Win(\C{G}_{b})$. 

\medskip

$(\Leftarrow)$ Suppose that  $q^*=\delta_b(q_0^b, \pi \cdot (X_{n+1}\cup Y_{n+1}))\! \in Win(\C{G}_{b})$. 
By following the strategy of the transducer $\C{T}_{\C{G}_b}$ from this point on, 
we are guaranteed to produce a trace that is accepted by
the automaton $\autval{b}$. 
More precisely, consider the strategy $\sigma$ defined as follows:
\begin{itemize}
\item $\sigma(X_1' \ldots X_k')=Y_k$, if $X_1' \ldots X_k'=X_1 \ldots X_k$  ($1 \leq k \leq n+1$)
\item $\sigma(X_1' \ldots X_{n+1}' \ldots X_h')=\omega_b(\delta_b'(q^*, X_{n+1}' \ldots X_{h-1}'), X_h')$, if 
$X_1' \ldots X_{n+1}'=X_1 \ldots X_{n+1}$ 
\item $\sigma(X_1') = \endtrace$, if $X_1' \neq X_1$
\item $\sigma(X_1' \ldots X_k' X_{k+1}' )= \endtrace$, if $X_i'=X_i$ for $1 \leq k< n+1$ and $X_{k+1}' \neq X_{k+1}$
\item $\sigma(X_1' \ldots X_h') = \emptyset$, in all other cases
\end{itemize}
where $\delta_b'$ and $\omega_b$ are respectively  the transition and output functions of the transducer $\C{T}_{\C{G}_b}$.
The first bullet concerns prefixes of $X_1 \ldots X_{n+1}$ and
ensures that $(X_1 \cup Y_1) \ldots (X_{n+1}\cup Y_{n+1}) \in \compat{\sigma})$.
The second bullet states that once $X_1 \ldots X_{n+1}$ has been read, 
we start following the transducer $\C{T}_{\C{G}_b}$. 
The remaining points ensure that all infinite traces $\pi'\in \inftraces{\sigma}$ such that $\pi \, (X_{n+1}\cup Y_{n+1})\not \sqsubseteq \pi'$ contain a single occurrence of $\endtrace$. 

We claim that $\sigma$ is an $\assume$-strategy such that 
$(X_1 \cup Y_1) \ldots (X_{n+1}\cup Y_{n+1}) \in \compat{\strategy}$ and $\wcvstart{\strategy}{\trace,X_{n+1}} \geq b$. 
As noted above, the first bullet of the definition ensures that  $(X_1 \cup Y_1) \ldots (X_{n+1}\cup Y_{n+1}) \in \compat{\strategy}$. 
The last three bullets make sure that every  $\pi' \in \inftraces{\sigma}$ with 
$(X_1 \cup Y_1) \ldots (X_{n+1}\cup Y_{n+1}) \not \sqsubseteq \pi'$ contains exactly one occurrence of $\endtrace$.
It remains to consider the infinite traces that begin with $(X_1 \cup Y_1) \ldots (X_{n+1}\cup Y_{n+1})$.
Consider some such trace $\pi' = (X_1 \cup Y_1) (X_2 \cup Y_2) \ldots \in \inftraces{\sigma}$, and let $s_0 s_1 s_2 \ldots$ be the infinite run of $\autval{b}$ on $\pi'$. 
Because $(X_1 \cup Y_1) \ldots (X_{n+1}\cup Y_{n+1}) \sqsubseteq \pi'$, we know that $s_{n+1} = q^*$. %
The latter can be combined with the second bullet to show that 
for every $i\geq n+1$,  $Y_{i+1} = \omega_b(s_i, X_{i+1})$ and $s_{i+1} = \delta_b'(s_i, X_{i+1})$.
As $s_{n+1}=q^* \in Win(\C{G}_{b})$, and $\omega_b$ is defined so as to always remain within $Win(\C{G}_{b})$,
we have $s_i \in Win(\C{G}_{b})$ for every $i\geq n+1$. Furthermore, since $\omega_b$ always reduces the distance to $RA(\C{G})$, 
there are infinitely many $i \geq n+1$ such that $s_i \in F_b$. 
This establishes that $s_0 s_1 s_2 \ldots$ is an accepting run, and thus, $\pi'$ must either be such that $\pi' \not \models \assume$,
or it contains a single occurrence of $\endtrace$ such that the finite induced trace $\pi''$  is such that $\eval{\pi''}{\varphi}\ \geq b$. 
We have thus shown that every infinite trace produced by $\sigma$ either violates $\assume$ or contains $\endtrace$, i.e.\ $\sigma$ is an $\assume$-strategy. 
Moreover, for every $\pi'' \in \fulltraces{\sigma}$ with $(X_1 \cup Y_1) \ldots (X_{n+1}\cup Y_{n+1}) \sqsubseteq \pi''$, we have $\eval{\pi''}{\varphi}\ \geq b$, as desired. 
\end{proof}

In the following lemma, we use $\sigma_\mathsf{str}$ to denote the strategy implemented by the transducer $\C{T}^\mathsf{str}$. To simplify the formulation, we extend the notation $\wcvstart{\strategy}{\trace,X}$ to all finite traces that can be produced by strategy $\strategy$ (recall that Definition 2 only defines this notation for compatible traces that do not contain $\endtrace$). Formally, given a  trace $\trace = (X_1 \cup Y_1) \ldots (X_n \cup Y_n)$ that is a prefix of some infinite trace in
$\inftracestrat{\strategy}{\mathbf{X}}$ and contains $\endtrace$ at position $k \leq n$, we set $\wcvstart{\strategy}{\trace,X}= \eval{\trace'}{\varphi}{}$, 
where $\pi'= (X_1 \cup Y_1) \ldots (X_k \cup Y_k \setminus \{\endtrace\})$.

\begin{lemma}\label{lem:strong2}
For every $X_1 \ldots X_n \in (2^\C{X})^*$, $X_{n+1} \subseteq \C{X}$, and $v \in [0,1]$: 
if %
$\pi= (X_1\cup  \omega_\mathsf{str}(q_0^\mathsf{str}, X_1)) \ldots (X_n \cup  \omega_\mathsf{str}(q_0^\mathsf{str}, X_1 \ldots X_n))$,
$(q_1, q_2, \ldots, q_m)=\delta_\mathsf{str}(q_0^\mathsf{str}, X_1 \ldots X_n)$, and 
$u$ is such that $\omega_\mathsf{str}((q_1, q_2, \ldots, q_m), X_{n+1})$ was set equal to  $\omega_u(q_u,X_{n+1})$, 
then $\wcvstart{\strategy_\mathsf{str}}{\pi ,X_{n+1}} \geq u$.  
\end{lemma}
\begin{proof}
Fix $X_1 \ldots X_n \in (2^\C{X})^*$ and $X_{n+1} \subseteq \C{X}$.
Consider some $\mathbf{X} \in (2^\C{X})^\omega$ such that $\mathbf{X} = X_1 \ldots X_n X_{n+1} X_{n+2} \ldots $.
Let $\pi' = \inftracestrat{\sigma_\mathsf{str}}{\mathbf{X}}=(X_1 \cup Y_1) (X_2 \cup Y_2) \ldots$, and for all $i \geq 0$ let $(q_{i+1}^1, \ldots, q_{i+1}^m) = \delta_\mathsf{str}(q_0^\mathsf{str}, X_1 \ldots X_{i+1})$. 
Define a function $\zeta: \mathbb{N} \rightarrow \{b_1, \ldots, b_m\}$ by letting $\zeta(i)$ be  the unique value $b$ in $\{b_1, \ldots, b_m\}$
such that  $\omega_\mathsf{str}((q_{i}^1, \ldots, q_{i}^m) , X_1 \ldots X_i), X_{i+1})$ was set equal to $\omega_b(q_i^b,X_{i+1})$. 
Basically, $\zeta(i)=b$ means that we used transducer $\C{T}_{\C{G}_b}$ to generate $Y_{i+1}$ after reading $X_{i+1}$ from state $(q_{i}^1, \ldots, q_{i}^m)$.
Note that the value $u$ from the lemma statement is equal to $\zeta(n)$. \smallskip

\noindent\textbf{Claim}: For all $i \geq 0$, $\zeta(i+1) \geq \zeta(i)$. \\
\emph{Proof of claim}: Suppose that  $\zeta(i)=b$, which means that $b= \mathsf{max}(\{v \in B\mid \exists Y\, \delta_v(q_i^v,X_{i+1} \cup Y) \in Win(\C{G}_{v})\})$.
Then we have $Y_{i+1} = \omega_b(q_i^b,X_{i+1})$. From the definition of $\omega_b$, and the fact that there exists some $Y$ with $\delta_b(q_i^b,X_{i+1} \cup Y) \in Win(\C{G}_{b})$,
we know that $\delta_b(q_i^b,X_{i+1} \cup Y_{i+1}) \in Win(\C{G}_{b})$. It follows that there must exist $Y'$ such that  
$\delta_b(q_{i+1}^b,X_{i+2} \cup Y') \in Win(\C{G}_{b})$, and hence $\mathsf{max}(\{v \in B\mid \exists Y'\, \delta_v(q_{i+1}^v,X_{i+2} \cup Y') \in Win(\C{G}_{v})\}) \geq b$.
Thus, $\zeta(i+1) \geq b = \zeta(i)$. (\emph{end proof of claim})\smallskip

Due to the preceding claim, and the finiteness of $B$, 
there exists $h \geq 0$ such that $\zeta(i)=\zeta(h)$ for all $i \geq h$, and $\zeta(i) < \zeta(h)$ for $i < h$. 
Let $b_\mathbf{X} = \zeta(h)$. Observe that $q_0^{b_\mathbf{X} } q_1^{b_\mathbf{X} } q_3^{b_\mathbf{X} } \ldots $ is the run of $\autval{b_\mathbf{X} }$ 
on $\pi'$ and that $q_{h+1}^{b_\mathbf{X}} \in Win(\C{G}_{b_\mathbf{X}})$. 
From the definition of $\omega_{b_\mathbf{X}}$, we can infer that $q_i^{b_\mathbf{X}} \in Win(\C{G}_{b_\mathbf{X}})$
for all $i > h$, and further that there are infinitely many $q_i^{b_\mathbf{X}}  \in F_{b_\mathbf{X}}$. 
This means that $q_0^{b_\mathbf{X} } q_1^{b_\mathbf{X} } q_3^{b_\mathbf{X} } \ldots $ is an accepting run of $\autval{b_\mathbf{X} }$ on $\pi'$.
Thus, either $\pi'$ does not contain $\endtrace$ and violates $\assume$, or $\pi'$ contains $\endtrace$ 
and the induced finite trace $\tracestrat{\sigma_\mathsf{str}}{\mathbf{X}}$ has value at least $b_\mathbf{X} = \zeta(h) \geq \zeta(n)=u$.  
We have shown that every trace $\pi'' \in \fulltraces{\sigma_\mathsf{str}}$ is such that $\eval{\pi''}{\varphi}\ \geq \zeta(n)=u$,
and hence that $\wcvstart{\strategy_\mathsf{str}}{\pi ,X_{n+1}} \geq u$.
\end{proof}

We now proceed to the proof of  Theorem \ref{thm:strong-opt}.\medskip

\noindent\textbf{Theorem \ref{thm:strong-opt}}
$\C{T}^\mathsf{str}$ implements a strongly \opt strategy and can be constructed in double-exponential time. %
\begin{proof}
We first show that $\C{T}^\mathsf{str}$ implements a strongly \opt strategy. Suppose for a contradiction that this is not the case. 
Then there exists a finite trace $\trace = (X_1 \cup Y_1) \ldots (X_n \cup Y_n) \in \compat{\strategy}$ that does not contain $\endtrace$, $X \in \mathbf{X}$, and an $\assume$-strategy
such that $\trace \in \compat{\strategy'}$ and $\wcvstart{\strategy'}{\trace,X} > \wcvstart{\strategy}{\trace,X}$.
Suppose that $\wcvstart{\strategy}{\trace,X} = b$ and $\wcvstart{\strategy'}{\trace,X}=b'$. 
From Lemma \ref{lem:strong1}, we know that $\delta_{b'}(q_0^{b'}, \pi \, (X_{n+1}\cup Y_{n+1}))\! \in Win(\C{G}_{b'})$
for some $Y_{n+1}\! \in 2^{\C{Y} \cup \{\endtrace\}}$. 
It follows that $\omega_\mathsf{str}(\delta_\mathsf{str}(q_0^\mathsf{str}, X_1 \ldots X_n), X_{n+1})= \omega_v(q_v,X_{n+1})$
for some $v \geq b'$. From Lemma \ref{lem:strong2}, $\wcvstart{\strategy_\mathsf{str}}{\pi ,X_{n+1}} \geq v \geq b'$, a contradiction.  

We know from Section \ref{ssec:algo-bgv} and Theorem \ref{thm:opt_strategies} that for each $v \in V(\varphi)$, 
the DBA $\autval{v}$, winning region $Win(\C{G}_{v})$, and transducer $\C{T}_{\C{G}_v}$ can be constructed in double-exponential time. 
As there are only single exponentially many values in $V$ (hence $B$), 
the transducer $\C{T}^\mathsf{str}$ can also be constructed double-exponential time.
\end{proof}

\end{document}